\documentclass[aps,pra,twocolumn,superscriptaddress]{revtex4-2}
\usepackage[utf8]{inputenc}
\usepackage{amsmath, amsthm, amssymb, amsfonts, enumitem}
\usepackage{graphicx}
\usepackage{dsfont}
\usepackage{mathtools}
\usepackage{pgfplots}
\pgfplotsset{compat=1.18}
\usepgfplotslibrary{groupplots,statistics}
\usepackage{xcolor}
\usepackage[colorlinks=true, linkcolor=blue, urlcolor=black, citecolor=blue]{hyperref}
\usepackage{orcidlink}
\usepackage{physics}

\DeclareMathOperator{\spec}{spec}
\newcommand{\ot}{\otimes}
\newcommand{\one}{\mathds{1}}
\newcommand{\I}{\one}
\renewcommand{\a}{\alpha}

\newcommand{\N}{\mathds{N}}
\newcommand{\C}{\mathds{C}}
\newcommand{\cdk}{(\C^d)^{\otimes k}}

\newtheorem{theorem}{Theorem}
\newtheorem{observation}[theorem]{Observation}
\newtheorem{lemma}[theorem]{Lemma}
\newtheorem{result}[theorem]{Result}

\begin{document}
\title{
Entanglement measures and LOCC conversion criteria from Schur-Weyl duality
}

\author{Anton van Cleemput$^{\orcidlink{0009-0005-3744-3774}}$}
\affiliation{Sorbonne Université, CNRS, LIP6, F-75005 Paris, France}
\affiliation{Naturwissenschaftlich-Technische Fakult\"{a}t, Universit\"{a}t Siegen, Walter-Flex-Stra\ss e 3, 57068 Siegen, Germany}

\author{Otfried Gühne$^{\orcidlink{0000-0002-6033-0867}}$}
\affiliation{Naturwissenschaftlich-Technische Fakult\"{a}t, Universit\"{a}t Siegen, Walter-Flex-Stra\ss e 3, 57068 Siegen, Germany}

\author{Albert Rico$^{\orcidlink{0000-0001-8211-499X}}$}
\affiliation{Naturwissenschaftlich-Technische Fakult\"{a}t, Universit\"{a}t Siegen, Walter-Flex-Stra\ss e 3, 57068 Siegen, Germany}

\date{\today}
\begin{abstract}
% The goal of this paper is to impress Marco Túlio Quintino so Anton gets a good grade for his defense. (And a box of chocolate?)
The characterization of possible transformations between quantum states and the quantification of entanglement are central for understanding quantum resources. Explicit results are, however, scarce since they require complicated optimizations over transformation protocols or state decompositions. We introduce a method to optimize functions of the
entanglement spectrum of bipartite states over mixed-state decompositions with arbitrary accuracy. Our results are based on a convergent semidefinite programming hierarchy of optimizations over entanglement spectra using a natural polynomial basis and Schur-Weyl duality. Our approach can also be used
to determine non-polynomial measures, such as the entanglement of formation,
and to characterize the dimensionality of entanglement.
Moreover, extracting the decomposition allows us to formulate
criteria to determine state transformations under local operations
and classical communication (LOCC). In particular,
we provide a hierarchy of certificates of LOCC non-convertibility
from pure to mixed states in any dimension and a heuristic approach
for LOCC convertibility.
\end{abstract}

\maketitle
\section{Introduction}
Composite quantum systems are entangled when they cannot be described solely by their parts~\cite{epr}. Operationally, entanglement is a resource for quantum communication~\cite{distributed, eckert} and computation~\cite{one-way, shor} that cannot increase under local operations and classical communications (LOCC) between distant parties.
In isolated two-particle quantum systems, described by pure states, entanglement and LOCC transformations are completely characterized by
the Schmidt coefficients or the eigenvalues of the marginals~\cite{Nielsen1999MajLOCC, Vidal_2000}, also called
the entanglement spectrum. Therefore, learning this spectrum is
necessary and sufficient to understand the entanglement of the
system and the possible entanglement transformations between
states.

In practice, however, pure states become mixed due to decoherence with the environment~\cite{Zurek2003,BreuerPetruccione2002,NielsenChuang2010}. When this occurs, the spectrum of the marginals alone is not informative of the entanglement between them, because a mixed state admits multiple descriptions as a probabilistic ensemble of pure states. Consequently,
one needs to account for all such decompositions. This leads to an intractable optimization required to evaluate entanglement measures~\cite{Bennett1996EoF, Uhlmann1998, Uhlmann2010Roofs, Horodecki2009RMP}, even when full tomography is available.
Existing methods address each specific entanglement measure independently~\cite{Wootters1998, TerhalVollbrecht2000, VollbrechtWerner2001, Mintert2004Concurrence, Chen2005Concurrence, Chen2005EoF, Guhne2007Estimating, Eisert2007QuantWitness}. Moreover, while lower bounds can be found for polynomial measures~\cite{Toth2015CroofTP}, much less is known for the more general case~\cite{PlenioVirmani2007}. As a result, it is still a remarkable challenge to determine when entanglement transformations are possible with LOCC~\cite{Gour2005InfiniteConditions, Chitambar2014LOCC, Sauerwein_2018}.

Here we aim to provide a systematic method to learn the entanglement properties of mixed bipartite quantum states.
To this end we use the Schur-Weyl decomposition of the Hilbert space of multiple copies, which allows one
to learn entanglement spectra of pure states \cite{Alicki1988, Keyl2001EstimatingSpectrum} to estimate any polynomial entanglement measure for mixed states with a complete polynomial basis. This basis is particularly suited to detect entanglement dimensionality. We find upper and lower converging bounds, which determine, for instance, the amount of entanglement of the main families of bound entangled states. Using semidefinite and linear programming, we extend them to a convergent hierarchy to find any possible entanglement measure of mixed states.
The combination of these bounds leads us to a hierarchy of certificates to determine LOCC convertibility from pure to mixed and a heuristic approach towards convertibility.

\section{Entanglement spectrum}

Consider a bipartite pure state $\ket{\psi}_{AB} \in \C^{d_A} \otimes \C^{d_B} \cong \C^D$ with $D = d_Ad_B$. Such a state is separable if it can be written as $\ket{\psi}=\ket{a}\otimes\ket{b}$ and entangled otherwise. More generally, any pure bipartite state admits a \emph{Schmidt decomposition}
\begin{equation}\label{eq:SchmidtDec}
    \ket{\psi}_{AB} = \sum_{i=1}^r \sqrt{\a_i}\ket{a_i}\otimes\ket{b_i},
\end{equation}
where $\a_1\geq\a_2\geq\dots\geq\a_r>0$ and the sets $\{\ket{a_i}\}_i$ and $\{\ket{b_i}\}_i$ are orthonormal. The spectrum $\a=(\a_1,\dots,\a_{d_A})$ of the reduced density matrix $\rho_A=\tr_B\dyad{\psi}$ (with $\a_i=0$ for $i>r$) is called the {\em entanglement spectrum} of $\ket{\psi}_{AB}$, and the number $r$ of nonzero coefficients is known as the {\em Schmidt rank}, denoted $\mathrm{SR}(\psi)$.
By the Nielsen theorem~\cite{Nielsen1999MajLOCC}, a bipartite pure state $\ket{\psi}$ with entanglement spectrum $\a(\psi)$ can be converted by LOCC into another bipartite pure state $\ket{\phi}$ with entanglement spectrum $\a(\phi)$ if and only if $E_j(\psi)\geq E_j(\phi)$ for all $j=1,\dots,d_A-1$, where
\begin{equation}\label{eq:Vidal}
    E_j(\psi)=1-\sum_{i=1}^{j} \a_i(\psi),
\end{equation}
as defined in Ref.~\cite{Vidal_1999}. In this sense, the quantities $E_j(\psi)$ form a complete set of entanglement monotones, and all information about the entanglement of each state is encoded in the entanglement spectrum. Every pure-state entanglement monotone is a concave, unitarily invariant function of $\rho_A$~\cite{Vidal_2000}.

A pure-state entanglement measure $E$ can be extended to mixed states via the convex roof extension,
\begin{equation}\label{eq:CRoof}
    \widehat{E}(\rho) = \min_{\mathcal{D}(\rho)} \sum_i p_i E(\psi_i),
\end{equation}
where the minimization is taken over all decompositions
$\mathcal{D}(\rho)=\{p_i,\ket{\psi_i}\}$ of $\rho$ such that $\rho = \sum_i p_i \dyad{\psi_i}$. However, this optimization is generally intractable, even if the state is known.

\section{Decomposition of the Hilbert space}
Let us now explain facts about the Hilbert space of $k$ copies of a state $\ket{\psi}$, since this is the basis of our approach.
A general $k$-qudit quantum state is a vector in the Hilbert space $\cdk$. We consider two groups acting on this space: the \emph{symmetric group} $S_k$ of permutations $\pi\in S_k$ of $k$ elements, acting as
\begin{align}
    \pi (\ket{\psi_1} \ot \dots \ot \ket{\psi_k}) &= \ket{\psi_{\pi^{-1}(1)}} \ot \dots \ot \ket{\psi_{\pi^{-1}(k)}},
\end{align}
and the \emph{general linear group} $GL_d$ of invertible $d\times d$ matrices $g\in GL_d$, acting as
\begin{align}
g(\ket{\psi_1} \ot \dots \ot \ket{\psi_k}) &=g\ket{\psi_1} \ot \dots \ot g\ket{\psi_k}.
\end{align}
These two actions commute,
\begin{equation}\label{eq:PerUniCommute}
    \pi \, g = g \, \pi.
\end{equation}
The so-called \emph{Schur-Weyl duality}~\cite{fulton2013representation} states that under the action of these two groups, the Hilbert space $\cdk$ decomposes as
\begin{equation}
    \cdk = \bigoplus_{\substack{\lambda\vdash k \\ \ell(\lambda)\leq d}} \mathcal{S}^\lambda_k \ot \mathcal{U}^\lambda_d
\end{equation}
into irreducible subspaces $\mathcal{S}^\lambda_k$ and $\mathcal{U}^\lambda_d$, which carry the irreducible representations
of $S_k$ and $GL_d$, respectively, indexed by partitions $\lambda$. The dimensions of $\mathcal{S}^\lambda_k$ and $\mathcal{U}^\lambda_d$ are denoted as $f^\lambda$ and $s_\lambda(1,\dots,1)$, where $s_\lambda$ is the so-called Schur polynomial (see also below). Each partition $\lambda$ is represented as a Young diagram: a non-increasing vector of positive integers summing to $k$, whose length we denote by $\ell(\lambda)$.  We write $\lambda \vdash k$ to denote all the partitions $\lambda$ of $k$.

To each Young diagram $\lambda$ we can associate a so-called {\em isotypic projector} $\Pi^\lambda$ onto the subspace $\mathcal{S}^\lambda_k \ot \mathcal{U}^\lambda_d$. This projector is given by
\begin{align}
    \Pi^\lambda = \frac{f^\lambda}{k!}\sum_{\pi \in S_k} \chi^\lambda(\pi)\, \pi,
\end{align}
where $\chi^\lambda(\pi)$ is the character of the cycle type of $\pi$ in the irreducible representation $\lambda$. It follows that the isotypic projectors commute with both permutations and tensor powers of matrices,
\begin{equation}
    \pi\,\Pi^\lambda = \Pi^\lambda\,\pi\,,\quad g\,\Pi^\lambda = \Pi^\lambda\,g\,.
\end{equation}

\section{Polynomial optimization over mixed-state decompositions}
We are now ready to present the main tool that we use to address the problem of entanglement estimation and convertibility certification. The spectrum of a reduced state $\rho_A$ can be estimated via weak Schur sampling on $k$ copies of $\rho_A$~\cite{Alicki1988,Keyl2001EstimatingSpectrum}: the probability of measuring a Young diagram $\lambda \vdash k$ is given by
\begin{equation}
    \tr(\Pi^\lambda \rho_A^{\otimes k}) = f^\lambda s_\lambda(\spec\rho_A).
\end{equation}
It is known that the Schur polynomials $s_\lambda$ with $\lambda \vdash k$ and $\ell(\lambda) \leq d_A$ form a basis of the homogeneous symmetric polynomials $P$ of degree $k$ of the spectrum of $\rho_A$ \cite{MacdonaldSymmFunc}, and they expand into the monomial basis in the following way:
\begin{equation}
    s_\lambda(\spec\rho_A) = \sum_{\mu \vdash k} K_{\lambda\mu}m_\mu(\spec\rho_A),
\end{equation}
where $m_\mu$ is the monomial symmetric polynomial
\begin{equation}
    m_\mu(\a) = \sum_{\beta} \a_1^{\beta_1}\a_2^{\beta_2}\dots\a_{d_A}^{\beta_{d_A}}
\end{equation}
where the sum runs over all {\em distinct} permutations $\beta$ of $(\mu_1, \mu_2, \dots, \mu_{\ell(\mu)}, 0,\dots, 0) \in \N^{d_A}$, and $K_{\lambda\mu}$ is a constant known as the Kostka number~\cite{MacdonaldSymmFunc}.

Since $\tr(\rho_A) = 1$ implies $\tr(\rho_A)^{l} = 1$, any polynomial of degree $l < k$ can be lifted to a homogeneous polynomial of degree $k$. Thus, any symmetric polynomial of degree at most $k$ of the entanglement spectrum can be expressed as
\begin{equation}
    P(\spec\rho_A) = \sum_{\lambda\vdash k}c_\lambda s_\lambda(\spec\rho_A).
\end{equation}

In order to estimate any entanglement measure of bipartite states given by a symmetric polynomial, we are interested in the convex roof of these polynomials,
\begin{equation}\label{eq:PrhoABcroof}
    \widehat{P}(\rho_{AB}) = \min_{\mathcal{D}(\rho_{AB})} \sum_i p_i \tr(O_P\, \rho_{A,i}^{\otimes k}),
\end{equation}
where $\rho_{A,i}=\tr_B\dyad{\psi_i}$ and
\begin{equation}\label{eq:OP}
    O_P = \sum_{\lambda \vdash k} \frac{c_\lambda}{f^\lambda} \Pi^\lambda
\end{equation}
for any tuple of coefficients $\vec{c}:=(c_{\lambda_1},\dots,c_{\lambda_t})$, where $t$ is the number of partitions of $k$. $\widehat{P}$ is an entanglement monotone if $P$ is concave~\cite{Vidal_2000}. For any $P$, we always have unitary invariance since $P$ only depends on the spectrum.

\section{A complete hierarchy for entanglement estimation}
Our approach consists of using the spectrum estimation formulation above for mixed-state decompositions. To this end, we use that polynomial entanglement measures can be bounded using
optimizations over functions of several copies of a state~\cite{Toth2015CroofTP}.
That is, given a bipartite mixed state $\rho_{AB}$ acting on $\C^{d_A} \otimes \C^{d_B}$, to any decomposition in $\mathcal{D}(\rho_{AB})$ we associate the state
\begin{equation}\label{eq:omega}
    \omega_k = \sum_i p_i \dyad{\psi_i}^{\otimes k}
\end{equation}
acting on $(\C^{d_A} \otimes \C^{d_B})^{\otimes k}$.
The core idea is that a polynomial entanglement measure of degree not larger than $k$ can be written as a linear function on $\omega_k$. Then, optimizing over all decompositions amounts to optimizing over all $\omega_k$ with
the constraint that $\omega_k$ is symmetric and fully separable with respect to the $k$ copies. The separability constraint is difficult to implement, but it can be treated by a hierarchy of semidefinite approximations, similar to the separability criterion of symmetric extensions
\cite{doherty2002}.

Explicitly, we can compute a sequence
of lower bounds on Eq.~\eqref{eq:PrhoABcroof} as $\widehat{P}(\rho)\geq P_{{\rm LB}_k}(\rho)$ with the following semidefinite programming (SDP) relaxation,
\begin{align}
    P_{{\rm LB}_k}(\rho)=\min_{\Omega_k} \ &\tr\!\left(O_P^A \otimes \I^B \, \Omega_k\right)\label{eq:SDPcroof1}\\
    \text{s.t. } & \Omega_k \succeq 0, \nonumber\\
    & \Pi_\text{sym}^{AB,k} \, \Omega_k = \Omega_k\,\Pi_\text{sym}^{AB,k}  = \Omega_k,\nonumber\\
    & \tr_{2,\dots,k} (\Omega_k) = \rho, \nonumber\\
    & \Omega_k^{T_l} \succeq 0 \quad  \forall \, l \in \{1, \dots, \lfloor \tfrac{k}{2}\rfloor\}.\nonumber
\end{align}
Here, $O_P^A$ acts on $A_1,\dots,A_k$. When $\deg(P) = m < k$, $P$ can be lifted to degree $k$, or, equivalently, evaluated on $A_1,\dots,A_m$ as both give the same bound (Lemma~\ref{lem:lifting}). $\Pi_\text{sym}^{AB,k}$ is the projector onto the $k$-copy symmetric subspace of $(\C^{d_A}\otimes\C^{d_B})^{\otimes k}$ across copies, and $T_l$ denotes the partial transpose on the first $l$ copies. Since every $\omega_k$ is in the feasible set of this SDP at every level $k$, it gives a lower bound on the convex roof. This yields an SDP hierarchy of $k$-th order lower bounds. This hierarchy is complete, as shown in Appendix~\ref{app:ProofCongNonExtHier}.
\begin{result}\label{thm:ConvergenceSDPcroof}
The bound of Eq.~\eqref{eq:SDPcroof1} is saturated asymptotically, $\lim_{k\rightarrow\infty} P_{{\rm LB}_k}(\rho) = \widehat{P}(\rho)$.
\end{result}
The proof sketch is as follows: by the de Finetti theorem for Bose-symmetric states~\cite{caves2002unknown,Christandl_2007}, the $m$-copy marginal of any feasible $\Omega_k$ is asymptotically close to a mixture $\sum p_i \dyad{\phi_i}^{\otimes m}$ whose first system marginal is asymptotically close to $\rho$ as $k$ increases. By continuity of the convex roof, $P_{{\rm LB}_k} \geq \widehat{P}(\rho) - \delta_k$ where $\lim_{k\to \infty}\delta_k=0$. Moreover, the bound is nondecreasing in $k$ (Appendix~\ref{app:monotonicity}).

Since $\Omega_k$ is supported on the symmetric subspace $\text{Sym}^k(\C^{d_A}\otimes\C^{d_B})\subset(\C^{d_A}\otimes\C^{d_B})^{\otimes k}$, we can exploit this symmetry to reduce the matrix size of $\Omega_k$ from $D^{k}$ to $\dim \text{Sym}^k(\C^{d_A}\otimes\C^{d_B}) = \binom{k+D-1}{k}$. By projecting onto the range of $\rho$, we can further reduce this size to $\binom{k+r-1}{k}$, with $r=\operatorname{rank}\rho$. Moreover, the variable $\Omega_k$ can be taken real for real states. These symmetry reductions are explained in detail in Appendix~\ref{app:SymRed}.

Naturally, using the isotypic polynomial basis allows us to detect the entanglement dimension, i.e. the Schmidt number~\cite{Terhal_2000}, of a target mixed state. To see this, recall that $\Pi^\lambda\ket{e_1}\otimes\dots\otimes\ket{e_k}=0$ if fewer than $\ell(\lambda)$ distinct vectors of an orthonormal basis $\{\ket{e_i}\}$ appear in the product, and therefore $\tr\big(\Pi^\lambda_A\otimes\one_B \dyad{\psi}_{AB}^{\otimes k}\big) = 0$ if and only if $\ell(\lambda) > \operatorname{SR}(\psi_{AB})$~\cite{Christandl_2005,lovitz2026TreeTNtest}.
Therefore, if $P_{{\rm LB}_k}\neq 0$ up to the solver precision, then no pure-state decomposition of $\rho$ exists in which all states have Schmidt rank lower than $\ell(\lambda)$. This is formally described as follows:
\begin{result}
   If $P_{{\rm LB}_k}> 0$ is certified for $O_P=\Pi^\lambda$ with $\lambda\vdash k$, then $\rho$ has Schmidt number at least $\ell(\lambda)$.
\end{result}
As a simple example, we test this Schmidt number hierarchy on the state $ \rho = p \dyad{\phi^+} + (1-p)\dyad{01}$ acting on $\C^3 \otimes \C^3$, with $\ket{\phi^+}=\frac{1}{\sqrt3}\sum_{i=0}^{2}\ket{ii}$, studied in Ref.~\cite{Tavakoli_2024}, which has Schmidt number $3$ for $p > 0$, and we detect it from $p\approx 0.01$ at level $k=3$ with $\lambda = (1,1,1)$.
It is easy to see that the state $\rho$ has Schmidt number three for any $p > 0$ , since every pure state in a decomposition lies in the span of $\ket{\phi^+}$ and $\ket{01}$, and any such state  has Schmidt rank three.

\section{Simplified relaxation}

Notice that the isotypic projectors $\Pi^\lambda$ only act on Alice's subsystem copies. This leads us to ask whether the SDP can be simplified. In turn, here we present a simplification of the hierarchy above using a smaller variable that allows us to access larger systems, at the cost of weaker bounds at each level $P_{{\rm LB}_{k,m}} \leq P_{{\rm LB}_m}$ with $P_{{\rm LB}_{k,m}}$ defined in Eq.~\eqref{eq:marginals_SDP} below. For $m \geq k$, define the following extension acting on $(\C^{d_A})^{\otimes m}\otimes(\C^{d_B})^{\otimes k}$ (i.e., $m$ copies of $A$ and $k$ copies of $B$),
\begin{equation}\label{eq:omega_km}
    \omega_{k,m} := \sum_i p_i \dyad{\psi_i}^{\otimes k}\otimes \rho_{A,i}^{\otimes (m-k)}
\end{equation}
with $\rho_{A,i} = \tr_B(\dyad{\psi_i})$. We relax this operator with a positive semidefinite variable $\Phi_{k,m}$, and denote $\Phi_{k}:=\tr_{A_{k+1}, \dots, A_m}(\Phi_{k,m})$.
Then we can lower bound $\widehat{P}(\rho)$ as $P_{{\rm LB}_{k,m}}\leq \widehat{P}(\rho)$, where
\begin{align}\label{eq:marginals_SDP}
    P_{{\rm LB}_{k,m}}:=&\min_{\Phi_{k,m}} \tr(O_P^A \otimes \I^B \ \Phi_{k,m})\\
    \text{s.t. } & \Phi_{k,m}\succeq 0, \nonumber\\
    & \Pi_\text{sym}^{AB,k} \, \Phi_{k}\, \Pi_\text{sym}^{AB,k}= \Phi_{k},\nonumber\\
    & \tr_{A_2,\dots,A_m,B_2,\dots,B_k} (\Phi_{k,m}) = \rho_{AB}, \nonumber\\
    & \tr_{B_1,\dots,B_k}(\Phi_{k,m}) \text{ is permutation invariant},\nonumber\\
    & \Phi_{k,m}^{T_S} \succeq 0 \quad  \forall \, S.\nonumber
\end{align}
Here $O_P^A$ is the observable of Eq.~\eqref{eq:OP} for a polynomial of degree $m$, acting on $A_1,\dots,A_m$, and $T_S$ denotes the partial transpose on a set $S$ of registers, where transposing a full copy acts on both $A_i$ and $B_i$. Since $\Phi_{k,m}$ can be assumed invariant under permutations of $A_1B_1, \dots, A_kB_k$ and $A_{k+1}, \dots, A_m$, it suffices to impose partial transpose conditions on the tuples $(i, j)$ where each tuple means taking the partial transpose on the first $i$ full-copy registers and the first $j$ marginal registers. One such set of representatives is $\{(i,j):0\leq i\leq k,\ 0\leq j<\tfrac{m-k}{2}\}\ \cup\ \{(i,\tfrac{m-k}{2}):0\leq i\leq\lfloor k/2\rfloor\}$ where the second set is only present for even $m-k$ and the tuple $(0,0)$ is excluded.
For $k\geq\deg P$, the marginal $\Phi_k$ of any feasible $\Phi_{k,m}$ is feasible for Eq.~\eqref{eq:SDPcroof1} at level $k$, so that $P_{{\rm LB}_k}\leq P_{{\rm LB}_{k,m}}$ and this hierarchy is also complete in $k$. In fact, notice that for $k=m$ the hierarchy of Eq.~\eqref{eq:marginals_SDP} reduces to Eq.~\eqref{eq:SDPcroof1}.

The parameters $m$ and $k$ can be used for a tradeoff between performance and efficiency. In terms of memory usage, this hierarchy is not necessarily less costly at level $(k_1, m_1)$ with $k_1< m_1$ than it is at level $(k_2, m_2)$ with $k_2 = m_2$. This depends on the rank of the target state $\rho$ and its marginals (See App.~\ref{app:supportReduction}). Typically, level $(k_1, m_1)$ is less costly for full-rank states while $(k_2, m_2)$ is less costly for low-rank ones.

\begin{observation}
   For full-rank states, the simplified hierarchy achieves tighter bounds at the same computational cost.
\end{observation}
We show the memory cost of the hierarchy at different levels for states of different ranks in Fig.~\ref{fig:memory} in the Appendix.

\section{Approximating the optimal decomposition}\label{sec:extractingTheDecomposition}
Our next goal is to extract the decomposition $\{p_i, \ket{\psi_i}\}$ from the optimal solution $\Omega_k$ found in the SDP. This will provide us control of how tight our estimates are. Notice that if the relaxation is tight, we have $\Omega_k = \omega_k$. This is in general not the case, and therefore we write $\Omega_k = \omega_k + \Delta$, where $\Delta$ accounts for the gap of the relaxation and solver noise.
For $k \geq 3$, this decomposition is unique up to phases when the operators $\dyad{\psi_i}$ are linearly independent~\cite{Kruskal1977}. When the relaxation is tight and $\Omega_k = \omega_k$, we can then retrieve it exactly with Jennrich's algorithm~\cite{Harshman1970,Leurgans1993}. Since $\Omega_k$ is in general not of this form, we instead use the PARAFAC (or CANDECOMP) algorithm~\cite{Harshman1970,CarrollChang1970} from the TensorLy package~\cite{tensorly}, initialized with Jennrich's algorithm, to find a decomposition $\{p_j, \ket{\psi_j}\}$ that locally minimizes the Frobenius norm
\begin{equation}
    \Theta_k:=\Big\|\Omega_k - \sum_j p_j \dyad{\psi_j}^{\otimes k}\Big\|_F,
\end{equation}
with $\|X\|_F = \sqrt{\sum_{i,j} |X_{ij}|^2}$. Details are given in Appendix~\ref{app:tensor}.
We obtain a decomposition of the state $\sigma = \sum_j p_j \dyad{\psi_j}$, which is close to, but generally different from $\rho$. We turn it into an exact decomposition of $\rho$ by defining
\begin{equation}\label{eq:repair}
    \ket{\varphi_j} = \rho^{1/2}\sigma^{-1/2}\sqrt{p_j}\ket{\psi_j},
\end{equation}
where $\sigma^{-1/2}$ is taken on the support of $\sigma$. If $\operatorname{supp} \rho \subseteq \operatorname{supp} \sigma$, then $\sum_j \dyad{\varphi_j} = \rho^{1/2}\Pi_\sigma\rho^{1/2} = \rho$, with $\Pi_\sigma$ the projector onto the support of $\sigma$. Hence $\{q_j,\ket{\phi_j}\}$ with $q_j = \|\varphi_j\|^2$ and $\ket{\phi_j}=\ket{\varphi_j}/\|\varphi_j\|$ is a decomposition of $\rho$ and gives a valid upper bound on the convex roof.

As a comparison, we also compute upper bounds by gradient descent over ensemble decompositions, which can be unitarily parametrized. This result was first proven by Schrödinger \cite{schrodinger36}, but regularly rediscovered~\cite{jaynes57, gisin89, HJW,  kirkpatrick06} (App.~\ref{app:HJW}). Both upper bounds are comparable for the linear entropy and the Vidal monotone $E_1$, while gradient descent gives tighter upper bounds for the entanglement of formation (Figs.~\ref{fig:EoF}, \ref{fig:namedLinearEntropy} and~\ref{fig:Vidal}). However, the tensor extraction method remains useful in cases where an upper bound cannot be extracted otherwise. The main result concerning optimal decompositions can be summarized as follows.

\begin{result}\label{thm:extraction}
    Let $P$ be a symmetric polynomial in pure-state Schmidt coefficients.
    It is possible to extract the decomposition of $\Omega_k$ optimizing $P$. By doing so, we obtain an upper bound $P_{{\rm UB}_k}(\rho)\geq \widehat{P}(\rho)$ which converges, $\lim_{k\rightarrow\infty}P_{{\rm UB}_k}(\rho)=\widehat{P}(\rho)$.
\end{result}

The proof is shown in App.~\ref{app:extractionConv}, and for it we use the rank reduction step (App.~\ref{app:tensor}), an instance of Jennrich's algorithm~\cite{Harshman1970,Leurgans1993} and the repair step from Eq.~\eqref{eq:repair}. In practice, we cannot run the SDP at $k$ large enough to use Jennrich's algorithm and the extracted decomposition is always a heuristic one from PARAFAC.

We show the bounds on the convex roof of the linear entropy $S_L(\psi)=1-\tr\rho_A^2$ in Fig.~\ref{fig:linearEntropy}.

% Shared colors and styles of the figures. Every figure reads its data from figures/data/*.csv, written by
% experiments/make_figures.py (comma separated, header row, nan for missing values: pgfplots skips those points).
\definecolor{cA}{HTML}{2a78d6}\definecolor{cB}{HTML}{eb6834}\definecolor{cC}{HTML}{1baf7a}%
\definecolor{cD}{HTML}{4a3aa7}\definecolor{cE}{HTML}{c03b7a}\definecolor{cF}{HTML}{8a6d1f}\definecolor{cGray}{HTML}{8a8984}%
\definecolor{lvA}{HTML}{c6dbf5}\definecolor{lvB}{HTML}{93bbec}\definecolor{lvC}{HTML}{5e98e0}%
\definecolor{lvD}{HTML}{2a78d6}\definecolor{lvE}{HTML}{174a8c}\definecolor{lvF}{HTML}{1baf7a}%
%
% \namedbars{csv file}{ylabel}{title}: lower bounds of the named states at every level, divided by the reference `ref`
% (printed above each state at height `label_y`), with the two upper bounds as markers. Columns: see named_csv in
% make_figures.py. The state positions (`pos`) leave a gap at x = 10 between the PPT and the NPT states.
\providecommand{\namedbar}[3]{% fill, column, offset
  \addplot[ybar, bar width=0.13, bar shift=0pt, draw=none, fill=#1, area legend]
    table[x expr=\thisrow{pos}+#3, y expr=\thisrow{#2}/\thisrow{ref}] {\namedfile};}
\providecommand{\namedbars}[3]{%
  \def\namedfile{#1}%
  \begin{tikzpicture}
    \begin{axis}[width=\textwidth, height=0.29\textwidth, ymin=0, ymax=1.38, xmin=0.3, xmax=14.7,
      ylabel={#2}, ytick={0,0.25,0.5,0.75,1}, xtick=data, xticklabels from table={#1}{state},
      table/col sep=comma, xticklabel style={rotate=35, anchor=north east, font=\scriptsize, inner sep=1pt},
      tick label style={font=\footnotesize}, label style={font=\small}, title={#3}, title style={font=\small,
      yshift=8pt}, clip=false, legend style={font=\scriptsize, at={(0.5,-0.42)}, anchor=north, legend columns=4,
      draw=none, /tikz/every even column/.append style={column sep=6pt}}, legend cell align=left,
      ytick style={draw=none}, xtick style={draw=none}, ymajorgrids, grid style={cGray!25}]
      \addplot[draw=none, forget plot] table[x=pos, y expr=0] {#1};% the ticks (xtick=data)
      \namedbar{lvA}{lb12}{-0.325}\addlegendentry{LB $(1,2)$}
      \namedbar{lvB}{lb13}{-0.195}\addlegendentry{LB $(1,3)$}
      \namedbar{lvC}{lb22}{-0.065}\addlegendentry{LB $(2,2)$}
      \namedbar{lvD}{lb23}{0.065}\addlegendentry{LB $(2,3)$}
      \namedbar{lvE}{lb33}{0.195}\addlegendentry{LB $(3,3)$}
      \namedbar{lvF}{lb24}{0.325}\addlegendentry{LB $(2,4)$}
      \addplot[only marks, mark=diamond*, mark size=2.2pt, cB, mark options={draw=white, line width=0.4pt}]
        table[x expr=\thisrow{pos}-0.14, y expr=\thisrow{ub_extract}/\thisrow{ref}] {#1};
      \addlegendentry{UB: extracted ($k=m=3$)}
      \addplot[only marks, mark=x, mark size=2.4pt, black, thick]
        table[x expr=\thisrow{pos}+0.14, y expr=\thisrow{ub_hjw}/\thisrow{ref}] {#1};
      \addlegendentry{UB: HJW}
      \addplot[only marks, mark=none, nodes near coords={$\pgfplotspointmeta$}, point meta=explicit symbolic,
        nodes near coords style={anchor=south, font=\tiny, inner sep=1pt}, forget plot]
        table[x=pos, y=label_y, meta=ref_label] {#1};
      \draw[cGray] (axis cs:0.3,1) -- (axis cs:14.7,1);
      \draw[cGray, densely dashed] (axis cs:10,0) -- (axis cs:10,1.38);
      \node[font=\scriptsize, anchor=south] at (axis cs:5,1.38) {PPT (bound entangled)};
      \node[font=\scriptsize, anchor=south] at (axis cs:12.5,1.38) {NPT};
    \end{axis}
  \end{tikzpicture}}
\begin{figure}
  \centering
  \begin{tikzpicture}
    \begin{axis}[width=\columnwidth, height=0.72\columnwidth, xlabel={$a$}, ylabel={bounds on $\widehat{S_L}$},
      ymin=0, ymax=4.2e-3, scaled y ticks=base 10:3, tick label style={font=\footnotesize}, label style={font=\small},
      legend style={font=\scriptsize, at={(0.5,-0.3)}, anchor=north, draw=none, legend columns=3,
      /tikz/every even column/.append style={column sep=4pt}}, legend cell align=left, table/col sep=comma]
      \def\file{figures/data/linear_entropy_horodecki24.csv}
      \addplot[lvA, thick] table[x=a, y=lb12] {\file};\addlegendentry{LB $(1,2)$}
      \addplot[lvB, thick] table[x=a, y=lb13] {\file};\addlegendentry{LB $(1,3)$}
      \addplot[lvC, thick] table[x=a, y=lb22] {\file};\addlegendentry{LB $(2,2)$}
      \addplot[lvD, thick] table[x=a, y=lb23] {\file};\addlegendentry{LB $(2,3)$}
      \addplot[lvE, thick, densely dashed] table[x=a, y=lb33] {\file};\addlegendentry{LB $(3,3)$}
      \addplot[lvF, very thick, dotted] table[x=a, y=lb24] {\file};\addlegendentry{LB $(2,4)$}
      \addplot[only marks, mark=diamond, mark size=3pt, cB, thick] table[x=a, y=ub_extract] {\file};
      \addlegendentry{UB: extracted, $k=m=3$}
      \addplot[only marks, mark=x, black, thick, mark size=1.8pt] table[x=a, y=ub_hjw] {\file};\addlegendentry{UB: HJW}
    \end{axis}
  \end{tikzpicture}
  \caption{Convex roof of the linear entropy $S_L=1-\tr\rho_A^2$ for the Horodecki $2\times4$ family. Lower bounds are shown at different levels of the hierarchy. The markers show the upper bounds extracted from our method and HJW gradient descent.}
  \label{fig:linearEntropy}
\end{figure}
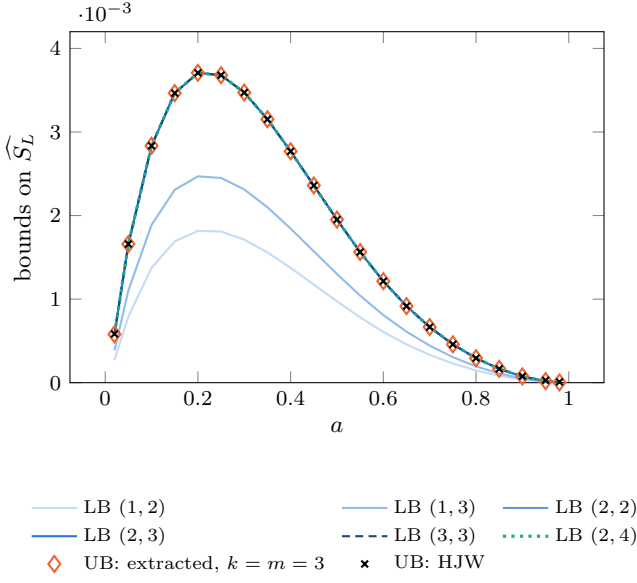

\section{Non-polynomial measures via Linear Programming}\label{sec:LP}

So far we have introduced a method to estimate polynomial entanglement measures, by using a natural basis of symmetric polynomials in the entanglement spectrum. Here we aim to use this approach to estimate non-polynomial entanglement measures. To this end notice that all continuous entanglement measures can be bounded arbitrarily well by infinite series of polynomials. Our goal is therefore to find the polynomial that best approximates the target measure. Namely, we aim to find
\begin{equation}\label{eq:FLBexact}
    \widetilde{F}_{{\rm LB}_k}(\rho)=\sup_{\substack{P\leq F\\ \deg P\leq k}}\ \min_{\Omega_k}\ \tr\left((O_P^A\otimes\I^B)\,\Omega_k\right),
\end{equation}
the optimization over symmetric polynomials of degree at most $k$ that lower bound $F$ on all spectra.

Using that all such polynomials can be written in the (discrete) basis of Schur polynomials, we will show that finding the best polynomial that lower bounds any entanglement measure is a problem that can be cast as a linear program (LP) over the coefficients in this basis. This allows us to dynamically find the best lower-bounding polynomial for a given state while running the SDP.

The proof is constructive, namely we will present how to construct such LP. First, we select multiple spectra $\a=(\a_1, \dots, \a_{d_A})$ at which we compute all the Schur polynomials $s_\lambda(\a)$ of degree $k$ and the input function $F(\a)$. These are points of the simplex of possible spectra $\alpha$, and define a grid $G=\{\a^g\}$ discretizing this simplex of possible spectra (Fig.~\ref{fig:grid}). Here the superindex $\a^g$ is chosen to distinguish the grid point $g$ from the subindex $j$ denoting each corresponding eigenvalue $\a^g_j$. We then require $F(\a^g) \geq \sum_\lambda c_\lambda s_\lambda(\a^g)$ at every spectrum $\a^g\in G$. By symmetry of the Schur polynomials, we only need to compute the points where $\a_i \geq \a_j$ for $i < j$ for each spectrum $\a$.
Thus, we are left with the following optimization problem,
\begin{equation}\label{eq:FLBgrid}
    F_{{\rm LB}_k}(\rho)=\max_{\substack{P(\a^g)\leq F(\a^g)\,\forall g\\ \deg P\leq k}}\ \min_{\Omega_k}\ \tr\left((O_P^A\otimes\I^B)\,\Omega_k\right).
\end{equation}

In order to solve this optimization, we need to address two challenges. First, a certain error $\epsilon_k$ is inherent to this grid procedure, in the sense that $|F_{{\rm LB}_k}-\widetilde{F}_{{\rm LB}_k}|\leq \epsilon_k$ (Fig.~\ref{fig:grid}).
Since we can retrieve the polynomial as output from the SDP, we can compute this error numerically and subtract it from the output to obtain a lower bound, namely $\widetilde{F}_{{\rm LB}_k}\geq F_{{\rm LB}_k} - \epsilon_k$. Second, this is in principle a min-max optimization problem. However, since the LP is feasible and has a bounded solution, strong duality holds, and we can turn this problem into a single minimization problem by taking the dual of the LP. Details are given in App.~\ref{app:LP}.

\begin{result}\label{thm:non-polyHier} We obtain a complete hierarchy to estimate any entanglement measure on mixed states, in the following sense:
\begin{enumerate}
    \item Let $F$ be a continuous symmetric function of the entanglement spectrum, and let $\widehat{F}(\rho)$ be its convex roof~\eqref{eq:CRoof}. Then $\widetilde F_{{\rm LB}_k}(\rho)\leq\widehat F(\rho)$ at all levels $k$.
    \item The bound is non-decreasing in $k$ and convergent, $\lim_{k\to\infty}\widetilde F_{{\rm LB}_k}(\rho)=\widehat F(\rho)$.
    \item With a finite grid $G$ of spectra $\alpha^g\in G$, $F_{{\rm LB}_k}(\rho)-\epsilon\leq\widehat F(\rho)$, where $\epsilon=\max_{\a^g}\left(P^*(\a^g)-F(\a^g)\right)$ is the grid error of the polynomial $P^*$ returned by the solver.
\end{enumerate}
\end{result}
\begin{figure}
  \centering
  \begin{tikzpicture}
    \def\file{figures/data/named_EF_main.csv}
    \begin{axis}[width=\columnwidth, height=0.62\columnwidth, ymin=0, ymax=1.16, xmin=0.4, xmax=8.6,
      ylabel={$E_{F,\mathrm{LB}}\,/\,\mathrm{UB}$}, ytick={0,0.25,0.5,0.75,1}, xtick=data,
      xticklabels from table={\file}{state}, table/col sep=comma,
      xticklabel style={rotate=35, anchor=north east, font=\scriptsize, inner sep=1pt},
      tick label style={font=\footnotesize}, label style={font=\small}, clip=false,
      ytick style={draw=none}, xtick style={draw=none}, ymajorgrids, grid style={cGray!25}]
      \addplot[ybar, bar width=0.55, bar shift=0pt, draw=none, fill=lvF]
        table[x=pos, y expr=\thisrow{lb_best}/\thisrow{ref}] {\file};
      \addplot[only marks, mark=none, nodes near coords={$\pgfplotspointmeta$}, point meta=explicit symbolic,
        nodes near coords style={anchor=south, font=\tiny, inner sep=1pt}]
        table[x=pos, y expr=1, meta=ref_label] {\file};
      \draw[cGray] (axis cs:0.4,1) -- (axis cs:8.6,1);
      \draw[cGray, densely dashed] (axis cs:6,0) -- (axis cs:6,1.16);
      \node[font=\scriptsize, anchor=south] at (axis cs:3,1.16) {PPT (bound entangled)};
      \node[font=\scriptsize, anchor=south] at (axis cs:7.5,1.16) {NPT};
    \end{axis}
  \end{tikzpicture}
  \caption{Best lower bound on the entanglement of formation $E_F$ for different PPT and NPT states, from the SDP with embedded LP at level $(k,m)=(2,4)$, divided by the best upper bound UB printed above each state. For two qubits, UB is the exact convex roof (marked $^*$). All levels of the hierarchy and both upper bounds, for more states, are shown in Fig.~\ref{fig:EoF}.}
  \label{fig:EoFmain}
\end{figure}
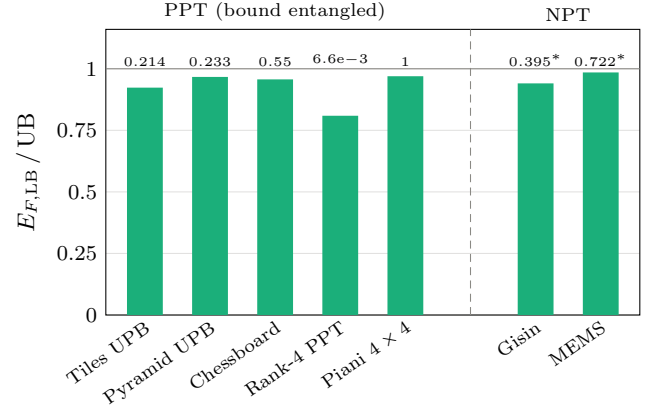

The proofs of all statements above are given in Appendix~\ref{app:LPconvergence}. We apply this method to estimate the convex roof of the von Neumann entropy of entanglement,
\begin{equation}\label{eq:vonNeumann}
    S(\ket{\psi_{AB}}) = -\tr(\rho_A \log_2 \rho_A),
\end{equation}
where $\rho_A = \tr_B \dyad{\psi_{AB}}$, extended to mixed states via convex roof, which we denote as $E_F=\widehat{S}$. This is known as the entanglement of formation~\cite{Bennett1996EoF}, and it is well known to be difficult to compute~\cite{Chen2005EoF}. To benchmark the method introduced, we use it to estimate $\widehat{S}(\rho_{AB})$ for several families of mixed states $\rho_{AB}$ of different local dimensions, including the main known bound entangled states. The results are given in Fig.~\ref{fig:EoFmain}, and in more detail in Fig.~\ref{fig:EoF} in the Appendix.

\begin{figure}[t]
    \centering
    \includegraphics{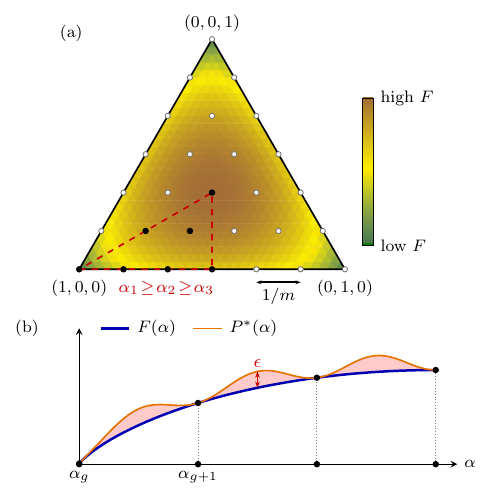}
    \caption{Grid discretization of the simplex of spectra for $d_A=3$.
    (a) Different values of $F(\alpha)$ are denoted by different colors. The dots denote the grid $G$, separated by a certain distance $1/m$. Since entanglement measures are symmetric functions on the spectrum, only the points with $\alpha_1\geq\alpha_2\geq\alpha_3$ (dashed region, filled dots) need to be imposed.
    (b) Cut along a one-parameter family of spectra. Although $P^*\leq F$ holds at all grid points, it is possible that $P^*>F$ between any two points by a value $\epsilon$. This value decreases as we refine the grid.}
    \label{fig:grid}
\end{figure}

%
% Classes (column `class` of locc.csv): 0 psi->rho certified (Theorem LOCC (ii)), 1 psi->rho excluded (Theorem LOCC (i)),
% 2 inconclusive, 3 incomparable (only if SHOW_REVERSE in make_figures.py). 31 x 31 grid, rows ordered by p then theta.
\begin{figure}[tbp]
  \centering
  \begin{tikzpicture}
    \begin{axis}[width=0.92\columnwidth, height=0.8\columnwidth, xlabel={$\theta/\pi$}, ylabel={$p$},
      xmin=-0.00416668, xmax=0.254167,
      ymin=-0.0166665, ymax=1.01667, axis on top,
      colormap={locc}{color(0)=(cA!55) color(1)=(cB!60) color(2)=(cGray!45) color(3)=(cC!60)},
      colormap access=direct,
      tick label style={font=\footnotesize}, label style={font=\small}, name=ax]
      \addplot[matrix plot*, mesh/cols=31, point meta=explicit, forget plot]
        table[col sep=comma, x=theta_over_pi, y=p, meta=class] {figures/data/locc.csv};
    \end{axis}
    \node[anchor=south, font=\scriptsize, inner sep=1pt] at ([yshift=3pt]ax.north) {%
      \begin{tabular}{@{}l@{\hspace{1.2em}}l@{}}
      \tikz\fill[cA!55] (0,0) rectangle (2mm,2mm); $\psi\to\rho$ (decomposition) &
      \tikz\fill[cB!60] (0,0) rectangle (2mm,2mm); $\psi\nrightarrow\rho$ (SDP) \\
      \tikz\fill[cGray!45] (0,0) rectangle (2mm,2mm); inconclusive
      \end{tabular}};
  \end{tikzpicture}
  \caption{LOCC convertibility from the pure state $\ket{\psi}$ with entanglement spectrum
  $(0.6,0.3,0.1)$ to the two-qutrit states
  $\rho(p,\theta)=p\dyad{\Phi_3}+(1-p)\dyad{\chi_\theta}$, with $\ket{\Phi_3}$ maximally entangled and
  $\ket{\chi_\theta}=\cos\theta\ket{00}+\sin\theta\ket{11}$, from the SDP of Eq.~\eqref{eq:LOCCminmaxt} at level $(k,m)=(3,4)$.}
  \label{fig:locc}
\end{figure}
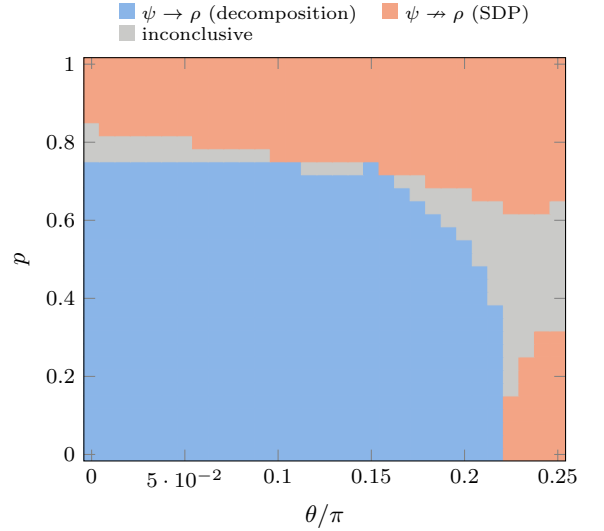

\section{LOCC convertibility}
Finally, let us demonstrate how the presented approach allows us to determine whether certain pairs of states can be transformed by means of LOCC. To date, and to the best of our knowledge, the only criteria known for this problem are restricted to qubits and very special families of states~\cite{Chitambar2014LOCC,Gour2005InfiniteConditions}.

To this end, let us first recall the following criterion~\cite{Jonathan_1999}: a pure state $\ket{\psi}$ can be converted into an ensemble $\{p_i, \ket{\psi_i}\}$ by LOCC, if and only if
\begin{equation}\label{eq:JPcriterion}
    E_j(\psi) \geq \sum_i p_i E_j(\psi_i)\qquad \text{for all } j=1,\dots,d_A-1,
\end{equation}
where $E_j$ are the Vidal monotones~\eqref{eq:Vidal}.
Hence, $\ket{\psi}$ can be converted into a mixed state $\rho$ if and only if some decomposition of $\rho$ satisfies all these inequalities simultaneously.

Notice that this is not straightforward to verify, since the decompositions of $\rho$ maximizing each monotone $E_j$ might be different. Therefore, to assure that all inequalities hold for the same decomposition (even if not optimal), we combine the LPs of Sec.~\ref{sec:LP} for all the $E_j$ into a single SDP. With a grid of spectra $\{\a^g\}_g$, we define the following optimization:
\begin{align}\label{eq:LOCCminmaxt}
    s_k=\min_{\Omega_k} \max_t \max_{P_j\in\mathcal{P}_j}\ \sum_j t_j\Big(\tr\left((O_{P_j}\otimes\I^B)\Omega_k\right)-E_j(\psi)\Big),
\end{align}
where $t$ is a probability vector and $\mathcal{P}_j$ is the set of polynomials of degree $k$ satisfying the grid constraint.

The optimal value $s_k$ is the smallest worst-case violation of Eq.~\eqref{eq:JPcriterion} over the set of relaxations of possible decompositions. The dual variables of this SDP provide weights $t_j^*\geq0$ with $\sum_jt^*_j=1$ and polynomials $Q_j^*$ satisfying $Q^*_j\leq t^*_jE_j$ on the grid (Appendix~\ref{app:LOCC}).
As shown in Sec.~\ref{sec:extractingTheDecomposition}, we can extract the decomposition from the result of the SDP. This leads us to a certificate that certain pure-to-mixed LOCC transformations are possible via Eq.~\eqref{eq:JPcriterion}. This can be summarized as follows.

\begin{result}
By extracting a decomposition $\{p_i, \ket{\psi_i}\}$, we can verify whether it satisfies Eq.~\eqref{eq:JPcriterion} for a certain input state $\ket{\psi}$, and therefore determine whether they can be transformed with LOCC.
\end{result}
However, unlike for the bounds for polynomial and non-polynomial measures, it may not always be possible to extract the correct decomposition. This comes from the fact that there may not exist a decomposition of $\rho$ into linearly independent terms that satisfies the criterion of Eq.~\eqref{eq:JPcriterion}, which makes it very difficult for the algorithms of Sec.~\ref{sec:extractingTheDecomposition} to find the underlying decomposition.
Instead, a similar strategy allows us to obtain a complete no-go criterion to discard entanglement transformations.

\begin{result}\label{thm:LOCC}
We obtain a complete hierarchy for LOCC non-convertibility in the following sense:
    \begin{enumerate}
    \item Let $\epsilon_j=\max_{\a^g}\left(Q^*_j(\a^g)-t^*_jE_j(\a^g)\right)$ be the grid errors of the polynomials returned by the solver.
    If $s_k>\sum_j\epsilon_j$ for some $k$, then $\ket{\psi}$ cannot be converted into $\rho$ by LOCC.
    \item In the limit where the grid $G$ is continuous and covers all possible spectra $\alpha$, for any $\ket{\psi}$ and $\rho$, if $\ket{\psi} \overset{\text{LOCC}}{\nrightarrow} \rho$, there is a finite level $k$ at which $s_k > 0$.
    \end{enumerate}
\end{result}
The proof is given in Appendix~\ref{app:LOCC}. Notice that although the grid will always be discrete in practical computations, by continuity the second item suggests that finer grids lead to tighter criteria. Numerically, we verify that this is the case for all states tested.

\section{Conclusions}

Using spectral estimation techniques and semidefinite programming (SDP), we have introduced a method to estimate any entanglement measure for bipartite mixed states.
This method also allows us to address the problem of determining whether entanglement transformations are possible from pure to mixed states, with a hierarchy of certificates of non-convertibility and a heuristic approach for convertibility.
This addresses a major barrier at certifying  quantum resources and their convertibility, and demonstrates the applicability of Schur-Weyl duality in this area.

For future research, it would be central to extend our approach to the multipartite case, where not all entanglement monotones are determined by the spectrum of the
reduced states. The classification of convertibility between multipartite states
is extremely difficult~\cite{Sauerwein_2018}, but at the same time highly relevant.
The reason is that relevant quantum phenomena like magic (for quantum computation) or the Fisher information (for quantum metrology) are inherently linked to multipartite systems. Extending our approach to these scenarios could be an essential step towards a full understanding of quantum resources.

\bigskip

{\em AI statement}.
The ideas, methodology,
proof strategies and the project scope
were designed by the authors.
The numerical results, plots and proof calculations were assisted with AI (Claude Opus 4.8-5.5), and validated by the authors. The article was entirely written by the authors, with minor refinements with AI to correct spelling and condense some of the sentences.

{\em Acknowledgements.} We thank Sophia Denker, Kenneth Goodenough, Ladina Hausmann, Renato
Renner and Karol \.Zyczkowski for discussions.
This research was funded by the Deutsche Forschungsgemeinschaft (DFG, German Research Foundation, project number 563437167), the Sino-German Center for Research Promotion (Project M-0294), and the German Federal Ministry of Research, Technology and Space (Project BeRyQC, Grant No.\ 13N17292).

\newpage

\bibliographystyle{ieeetr}
\bibliography{Bibliography}

\begin{widetext}

\appendix

\section{Proof of Result~\ref{thm:ConvergenceSDPcroof}}\label{app:ProofCongNonExtHier}
We consider a polynomial $P$ of degree $m$. At level $k \geq m$ of the hierarchy, we obtain the optimal solution $\Omega_k$, with $P_{{\rm LB}_k}(\rho) = \tr(O_P\, \Omega_m)$ where $\Omega_m = \tr_{m+1,\dots, k}(\Omega_k)$ is the $m$-copy marginal of $\Omega_k$.
By the de Finetti theorem for Bose-symmetric states~\cite[Thm.~II.8]{Christandl_2007}, for every feasible $\Omega_k$, there is an $\omega_m = \int \dyad{\phi}^{\otimes m} d\mu(\phi)$ such that
\begin{equation}
    \left\|\Omega_m - \omega_m\right\|_1 \leq \epsilon
\end{equation}
with $\epsilon = \tfrac{4 D m}{k}$. By Carath\'eodory's theorem, since the space is finite-dimensional, we can write $\omega_m = \sum_i p_i \dyad{\phi_i}^{\otimes m}$. Here, $\omega_m$ is associated to a decomposition $\{p_i, \ket{\phi_i}\}_i$ corresponding to a state $\sigma$ and $\tr(O_P \,\omega_m) \geq \widehat{P}(\sigma)$. Since the partial trace is contractive, $\left\|\rho - \sigma\right\|_1 = \left\|\tr_{2,\dots, m}(\Omega_m - \omega_m)\right\|_1 \leq \epsilon$.

Then, by H\"{o}lder inequality \eqref{eq:HolderStep}, by definition of the convex roof \eqref{eq:croofstep} and from the continuity of the convex roof of a continuous function in finite dimensions \eqref{eq:contConv}~\cite[Thm.~6.1]{Bucicovschi2013},
\begin{align}
    \tr(O_P \Omega_m) &\geq \tr(O_P\omega_m) - \epsilon\|O_P\|_\infty \label{eq:HolderStep}\\
    &\geq \widehat{P}(\sigma) - \epsilon\|O_P\|_\infty \label{eq:croofstep}\\
    &\geq \widehat{P}(\rho) - \delta(\epsilon) - \epsilon\|O_P\|_\infty \label{eq:contConv}
\end{align}
with $\delta(\epsilon) \to 0$.

Hence, $\widehat{P}{(\rho)} \geq P_{{\rm LB}_k}(\rho) \geq \widehat{P}{(\rho)} -\delta(\epsilon) - \epsilon\|O_P\|_\infty$, and therefore $P_{{\rm LB}_k}(\rho) \to \widehat{P}(\rho)$ as $k \to \infty$.

\section{Lifting}
\begin{lemma}\label{lem:lifting}
    Let $P$ be a homogeneous symmetric polynomial of degree $m$ with $m < k$. Let $P_k$ be its lift to degree $k$ such that $P_k(\a) = P(\a) \left(\sum_i \a_i \right)^{k-m} = P(\a)$ where $\a$ is the spectrum of a state $\rho_A$.
    Then, for every $\Omega_k$ acting on $(\C^{d_A} \otimes \C^{d_B})^{\otimes k}$, which is permutation invariant across all $k$ copies of $\C^{d_A} \otimes \C^{d_B}$,
    \begin{equation}
        \tr(O_P^{A_1,\dots, A_m} \otimes \I^{A_{m+1}, \dots, A_k} \otimes \I^{B} \Omega_k) = \tr(O_{P_k}^A \otimes \I^B \Omega_k).
    \end{equation}
\end{lemma}
\begin{proof}
    Let $\mathcal{T} = \frac{1}{k!}\sum_{\pi \in S_k} \pi \ O_P \otimes \I^{\otimes (k-m)} \ \pi^\dagger$ be the twirl of $O_P \otimes \I^{\otimes (k-m)}$ over $S_k$. Obviously $\mathcal{T}$ commutes with every permutation in $S_k$, and, since it is a linear combination of permutations, we can write it as a sum of isotypic projectors $\mathcal{T} = \sum_{\lambda \vdash k} a_\lambda \Pi^\lambda$ for some coefficients $a_\lambda$.

    For any $\rho_A$ with spectrum $\a$, since $\rho_A^{\otimes k}$ is permutation invariant,
    \begin{equation}
        \tr(\mathcal{T}\rho_A^{\otimes k}) = \tr(\frac{1}{k!} \sum_{\pi \in S_k} \pi O_P \otimes \I \pi^\dagger \rho_A^{\otimes k})= \tr(O_P \otimes \I \rho_A^{\otimes k}).
    \end{equation}
    Thus $\mathcal{T}$ and $O_{P_k}$ have the same expectation value $P_k(\a)$ on every $\rho_A^{\otimes k}$, and since both are combinations of isotypic projectors and the Schur polynomials are linearly independent, $\mathcal{T} = O_{P_k}$.
    Finally,
    \begin{align}
        &\tr(O_P^{A_1,\dots, A_m} \otimes \I^{A_{m+1}, \dots, A_k} \otimes \I^{B} \Omega_k)\\
        =& \tr( \frac{1}{k!}\sum_{\pi \in S_k} O_P^{A_1,\dots, A_m} \otimes \I^{A_{m+1}, \dots, A_k} \otimes \I^{B} \pi^\dagger \Omega_k \pi)\\
        =& \tr( \frac{1}{k!}\sum_{\pi \in S_k}  \pi O_P^{A_1,\dots, A_m} \otimes \I^{A_{m+1}, \dots, A_k} \otimes \I^{B} \pi^\dagger \Omega_k)\\
        =& \tr( O_{P_k} \otimes \I^{B} \Omega_k)
    \end{align}
\end{proof}

\section{Monotonicity in the level}\label{app:monotonicity}
In the full SDP, at level $k+1$, we obtain the optimal solution $\Omega_{k+1}$. Let $O_k$ be the observable at level $k$ and $O_{k+1}$ be its lift to level $k+1$. Then, $P_{{\rm LB}_{k+1}} = \tr(O_{k+1} \Omega_{k+1})=\tr(O_k \otimes \I \ \Omega_{k+1})$ according to Lemma~\ref{lem:lifting}. At level $k$, $\tr_{A_{k+1}B_{k+1}}(\Omega_{k+1})$ is also in the feasible set -- PPT across copies, permutation invariant, first marginal equal to $\rho$ and PSD -- with $\tr(O_k \tr_{A_{k+1}B_{k+1}}(\Omega_{k+1})) = P_{{\rm LB}_{k+1}}$. Then, $P_{{\rm LB}_{k+1}} \geq P_{{\rm LB}_{k}}$.
We can use the same argument in $k$ and $m$ for the simplified SDP.

\section{Symmetry reductions and optimizations}\label{app:SymRed}
Throughout this appendix, a single copy of $\C^{d_A}\otimes\C^{d_B}$ is identified with $\C^D$, with orthonormal basis $\{\ket{j}\}_{j=0}^{D-1}$.

\subsection{Restriction to the symmetric subspace}
The symmetric subspace $\text{Sym}^k(\C^D)$ has an orthonormal basis indexed by occupation vectors. The occupation vector denoted $\nu(i_1,\dots,i_k)\in\N^D$ of a basis vector $\ket{i_1,\dots,i_k}\in(\C^D)^{\otimes k}$ has entries $\nu_j$ equal to the number of times $j$ appears in $(i_1,\dots,i_k)$. For example, $\nu(0,1,2,0,0,1,2,2)=(3,2,3)$ for $D=3$, $k=8$. We define the isometry $V_k:(\C^D)^{\otimes k}\to\text{Sym}^k(\C^D)$ as
\begin{equation}
    V_k = \sum_{i_1, \dots, i_k=0}^{D-1} \binom{k}{\nu(i_1, \dots, i_k)}^{-1/2}\ket{\nu(i_1, \dots, i_k)}\bra{i_1, \dots, i_k},
\end{equation}
where $\binom{k}{\nu}=k!/\prod_j\nu_j!$ is the multinomial coefficient, i.e., the number of basis vectors with occupation vector $\nu$. It has real entries and satisfies $V_kV_k^\dagger=\I_{\text{Sym}^k}$ and $V_k^\dagger V_k=\Pi_{\rm sym}^{k}$. Since every feasible $\Omega_k$ is supported on the symmetric subspace, $\Omega_k=V_k^\dagger\Omega_k^{\rm sym}V_k$ with $\Omega_k^{\rm sym}=V_k\Omega_kV_k^\dagger$, a matrix of size $\binom{k+D-1}{k}$. The objective becomes $\tr\left(V_k(O_P\otimes\I^B)V_k^\dagger\,\Omega_k^{\rm sym}\right)$.

\subsection{First marginal constraint}
Here we wish to encode the condition $\tr_{A_2B_2\dots A_kB_k}(\Omega_k) = \rho$.
We define the occupation vector $\nu_\downarrow$ such that $\ket{\nu_\downarrow} \in \text{Sym}^{k-1}(\C^D)$, and $\ket{\nu_\downarrow + e_j} \in \text{Sym}^k(\C^D)$ has one more element $\ket{j}$, i.e., $(\nu_\downarrow + e_j)_j=(\nu_\downarrow)_j + 1$. Grouping the $\binom{k-1}{\nu_\downarrow}$ tuples $(i_2,\dots,i_k)$ with the same occupation vector and using $\binom{k}{\nu_\downarrow+e_j}=\binom{k-1}{\nu_\downarrow}\frac{k}{\nu_{\downarrow,j}+1}$, we get
\begin{equation}
    \left(\tr_{2,\dots,k}\Omega_k\right)_{jj'}
    =\sum_{\nu_\downarrow}\binom{k-1}{\nu_\downarrow}\frac{\bra{\nu_\downarrow+e_j}\Omega_k^{\rm sym}\ket{\nu_\downarrow+e_{j'}}}{\sqrt{\binom{k}{\nu_\downarrow+e_j}\binom{k}{\nu_\downarrow+e_{j'}}}}
    =\frac1k\tr\!\left(\hat a_{j'}^\dagger\hat a_j\,\Omega_k^{\rm sym}\right),
\end{equation}
with the bosonic creation operators $\hat a_j^\dagger:\text{Sym}^{k-1}(\C^D)\to\text{Sym}^k(\C^D)$,
\begin{equation}
    \hat a_j^\dagger=\sum_{\nu_\downarrow}\sqrt{\nu_{\downarrow,j}+1}\,\ket{\nu_\downarrow+e_j}\bra{\nu_\downarrow}.
\end{equation}
Thus, the marginal constraint becomes the linear constraint $\frac1k\sum_{j,j'}\tr\left(\hat a_{j'}^\dagger\hat a_j\,\Omega_k^{\rm sym}\right)\ket{j}\bra{j'}=\rho$.

\subsection{Partial-transpose constraints}
Thanks to the permutation symmetry, the partial transpose only depends on the size of the set of transposed copies, and since $\Omega_k^{T_{S^c}}=(\Omega_k^{T_S})^T$, a set and its complement give equivalent constraints. We therefore impose the PPT constraints on the first $l$ copies for $l \in \{1, \dots, \lfloor \tfrac{k}{2} \rfloor\}$. We define the isometry
\begin{equation}
    W_l=(V_l\otimes V_{k-l})V_k^\dagger:\ \text{Sym}^k(\C^D)\to\text{Sym}^l(\C^D)\otimes\text{Sym}^{k-l}(\C^D),
\end{equation}
so that $\Omega_k=(V_l\otimes V_{k-l})^\dagger\,W_l\Omega_k^{\rm sym}W_l^\dagger\,(V_l\otimes V_{k-l})$. Writing $T_l$ also for the transpose of the factor $\text{Sym}^l(\C^D)$ in the occupation basis, the partial transpose becomes
\begin{align}
    \Omega_k^{T_l} &= \left((V_l\otimes V_{k-l})^\dagger\,W_l\Omega_k^{\rm sym}W_l^\dagger\,(V_l\otimes V_{k-l})\right)^{T_l}\\
    &=\left(V_l^T\otimes V_{k-l}^\dagger\right)\,\left(W_l\Omega_k^{\rm sym}W_l^\dagger\right)^{T_l}\,\left(V_l^*\otimes V_{k-l}\right)\\
    &= (V_l\otimes V_{k-l})^\dagger\,\left(W_l\Omega_k^{\rm sym}W_l^\dagger\right)^{T_l}\,(V_l\otimes V_{k-l}),
\end{align}
since $V_l$ is real, so that $V_l^T=V_l^\dagger$ and $V_l^*=V_l$. Conjugating back with $V_l\otimes V_{k-l}$, the condition becomes
\begin{equation}
    \Omega_k^{T_l}\succeq0\iff\left(W_l\,\Omega_k^{\rm sym}W_l^\dagger\right)^{T_l}\succeq0,\qquad 1\leq l\leq\lfloor k/2\rfloor.
\end{equation}
This is the most expensive constraint of the SDP, since it encodes the state in the larger space $\text{Sym}^l(\C^D)\otimes\text{Sym}^{k-l}(\C^D)$ to support the partial transpose. While these constraints are not necessary for the hierarchy to be complete, they help convergence at finite $k$, as shown in Ref.~\cite{Navascues2009}.

\subsection{Real variables}
If $\rho$ is real, the SDP can be restricted to real symmetric matrices. If $\Omega_k$ is feasible, so is its entrywise complex conjugate $\Omega_k^*$, since conjugation commutes with partial traces, partial transposes and the real projector $\Pi_{\rm sym}^k$, and $\rho^*=\rho$. Since the isotypic projectors are linear combinations of permutation matrices with real characters, $O_P$ is real, and
\begin{equation}
    \tr\left((O_P^A \otimes \I^B)\, \Omega_k\right) = \tr\left((O_P^A \otimes \I^B)\, \Omega_k^*\right) = \tr\Big((O_P^A \otimes \I^B)\, \tfrac{\Omega_k + \Omega_k^*}{2}\Big).
\end{equation}
By convexity, $\frac12(\Omega_k+\Omega_k^*)$ is feasible, real and has the same objective value, so the optimum is attained by a real matrix. For the true $\omega_k$, this real matrix is $\frac{1}{2}\sum_i p_i \left(\dyad{\psi_i}^{\otimes k} + \dyad{\psi_i^*}^{\otimes k}\right)$, which corresponds to the decomposition $\{\frac{p_i}{2},\ket{\psi_i}\}\cup\{\frac{p_i}{2},\ket{\psi_i^*}\}$ of $\rho$. The same argument applies to $\Phi_{k,m}$ in the simplified relaxation.

\subsection{Projection onto the support of $\rho$}\label{app:supportReduction}
Let $r = \operatorname{rank} \rho$, and let $V_\rho: \C^D \to \C^r$ be the isometry onto the support of $\rho$, given by the eigenvectors of $\rho$ with nonzero eigenvalue, so that $V_\rho V_\rho^\dagger=\I_r$ and $V_\rho^\dagger V_\rho=\Pi_\rho$, the projector onto the support of $\rho$. Every feasible $\Omega_k$ is supported on $(\operatorname{supp}\rho)^{\otimes k}$: since
\begin{equation}
    \tr\left(((\I-\Pi_\rho)\otimes\I)\,\Omega_k\right)=\tr\left((\I-\Pi_\rho)\rho\right)=0
\end{equation}
and $\Omega_k\succeq0$, we have $((\I-\Pi_\rho)\otimes\I)\,\Omega_k=0$, and by permutation symmetry the same holds on every copy. Hence $\Omega_k=\Pi_\rho^{\otimes k}\Omega_k\Pi_\rho^{\otimes k}$, and
\begin{equation}
    \tr\left((O_P^A\otimes\I^B)\,\Omega_k\right)=\tr\Big(\left(V_\rho^{\otimes k}(O_P^A\otimes\I^B)V_\rho^{\dagger\otimes k}\right)\left(V_\rho^{\otimes k}\Omega_kV_\rho^{\dagger\otimes k}\right)\Big).
\end{equation}
We can therefore equivalently optimize over $\Omega_k^\rho=V_\rho^{\otimes k}\Omega_kV_\rho^{\dagger\otimes k}$ on $(\C^r)^{\otimes k}$ with marginal $V_\rho\rho V_\rho^\dagger$. Bose symmetry is preserved, since $V_\rho^{\otimes k}$ commutes with permutations, and the partial-transpose constraints are equivalent, since $(\Omega_k^\rho)^{T_l}$ and $\Omega_k^{T_l}$ are related in both directions by congruences with $V_\rho$ and $V_\rho^*$. Combined with the symmetric subspace, the variable has size $\binom{k+r-1}{k}$. If $\rho$ is real, $V_\rho$ can be chosen real and the real-variable reduction still applies.

The same argument applies to the simplified relaxation: by the permutation invariance of $\tr_B\Phi_{k,m}$, every $A$ register has reduced state $\rho_A=\tr_B(\rho)$, and therefore all $A$ registers can be encoded on the support of $\rho_A$.

\section{Tensor decomposition}
\subsection{Heuristic extraction}\label{app:tensor}
We take the optimal solution $\Omega_k$ of Eq.~\eqref{eq:SDPcroof1}. We work on the support of $\rho$ as in Appendix~\ref{app:supportReduction}: with $r=\operatorname{rank}\rho$ and the isometry $V_\rho:\operatorname{supp}\rho\to\C^r$, we consider $\Omega_\rho=V_\rho^{\otimes k}\Omega_kV_\rho^{\dagger\otimes k}$.

To work with real tensors, let $\{G_a\}_{a=0}^{r^2-1}$ be a Hilbert--Schmidt orthonormal basis of Hermitian operators on $\C^r$, consisting of the normalized identity $G_0=\I/\sqrt r$ and the generalized Gell-Mann matrices. Every pure state can be written as $\dyad{\psi}=\sum_ax_aG_a$ with a real vector $x$ satisfying $x_0=1/\sqrt r$ and $\|x\|_2=1$. We define the real symmetric order-$k$ tensor $\Omega^G_{a_1\dots a_k}=\tr\left((G_{a_1}\otimes\dots\otimes G_{a_k})\,\Omega_\rho\right)$, which for the true $\omega_k$ equals $\sum_ip_i\,x^{(i)\otimes k}$. We run PARAFAC (alternating least squares) on $\Omega^G$ with nonnegative weights, obtaining terms $\{p_i,v_i\}$; we normalize each factor such that $v_{i,0}=1/\sqrt r$, absorbing the scale into $p_i$, and recover the pure states $\ket{\psi_i}$ as the top eigenvectors of $\sum_av_{i,a}G_a$. Finally, we repair the decomposition $\{p_i,\ket{\psi_i}\}$ as described in Sec.~\ref{sec:extractingTheDecomposition}.

The number of terms $n$ fitted by PARAFAC is the rank of $\Omega_\rho$, which equals the number of terms when the $\dyad{\psi_i}^{\otimes k}$ are linearly independent.
Any decomposition of $\rho$ has at least $r$ terms, otherwise the repair fails. Conversely, every decomposition of $\rho$ consists of states on $\operatorname{supp}\rho\cong\C^r$, whose state space has dimension $r^2-1$, and therefore by Carath\'eodory's theorem the convex roof of a continuous function is attained by a decomposition with at most $r^2-1$ terms~\cite[Lemma~A-2]{Uhlmann1998}.

Before the extraction, it helps to reduce the rank of the SDP solution. We solve a second SDP, minimizing $\tr(R_3 \ot \I \,\Omega_k)$ for a random Hermitian $R_3$ acting on the first $3$ copies, over the same feasible set with the additional constraint $\tr\left((O_P\otimes\I^B)\Omega_k\right)\leq P_{\rm LB}+\delta_k$ for a small $\delta_k\to0$ with $\delta_k \geq \widehat{P}(\rho)-P_{{\rm LB}_k}$. This selects an extreme point of the near-optimal face, which generically has low rank, instead of a mixture of several optimal solutions.

\subsection{Convergence of the extraction}\label{app:extractionConv}
Given a state $\omega_3 = \sum_i p_i \dyad{\psi_i}^{\otimes 3}$, we can use Jennrich's algorithm to recover uniquely the decomposition $\{p_i, \ket{\psi_i}\}_i$ when the $\dyad{\psi_i}$ are linearly independent. Therefore, we need to select a low-rank extreme point $\Omega_3 = \tr_{A_4B_4\dots A_kB_k}(\Omega_k)$. For this we apply the rank reduction step as shown in App.~\ref{app:tensor}.

Now we need the following observation.
\begin{observation}
    A decomposition $\{q_i, \dyad{\psi_i}\}$ of $\rho$ is extremal in the set of optimal decompositions of $\rho$ with regards to a polynomial $P$, only if $\dyad{\psi_i}$ are linearly independent projectors.
\end{observation}
\begin{proof}
We will do so by contradiction as follows. Given an optimal decomposition $\{q_i, \dyad{\psi_i}\}$ of $\rho$, if the projectors $\dyad{\psi_i}$ are linearly dependent, there exists a nonzero real vector $c$ such that $\sum_i c_i \dyad{\psi_i} = 0$ with $\sum_i c_i=0$.
Assuming that this is the case, for a small $0<t\leq \min_i \tfrac{q_i}{|c_i|}$, let us define two other decompositions $\{q_i + t c_i, \dyad{\psi_i}\}$ and $\{q_i - t c_i, \dyad{\psi_i}\}$ such that their midpoint is the decomposition $\{q_i, \dyad{\psi_i}\}$.
Since $\sum_i (q_i+t c_i) P(\psi_i) \geq \widehat{P}(\rho)=\sum_i q_i P(\psi_i)$ and $\sum_i (q_i- t c_i) P(\psi_i) \geq \widehat{P}(\rho)$, then $\sum_i (q_i+ t c_i) P(\psi_i) = \sum_i (q_i- t c_i) P(\psi_i)= \widehat{P}(\rho)$.
These decompositions are also optimal and their midpoint is the decomposition extracted $\{q_i, \dyad{\psi_i}\}$, which proves that it is not an extreme point. Therefore, any extreme point of the convex set of optimal decompositions must have linearly independent projectors.
\end{proof}

For almost every $R$, the minimizer $\omega_3^*$ of $\tr(R\,\omega_3)$ among the $\omega_3$ of optimal decompositions is unique. The decomposition of $\rho$ giving $\omega_3^*$ has linearly independent projectors. Using the same argument shown above, if it has linearly dependent terms, we can define two decompositions whose associated $\omega_3$ is optimal with midpoint $\omega_3^*$. By uniqueness, both decompositions give $\omega_3^*$. We can simply increase $t$ until a weight vanishes and repeat to end up at a decomposition into linearly independent terms. Therefore, $\omega_3^*$ decomposes into linearly independent terms.

As $k\to\infty$ the rank-reduced $\Omega_3$ converges to $\omega_3(\mu_e)=\sum_ip_i\dyad{\psi_i}^{\otimes3}$, where
$\mu_e$ is the extreme optimal decomposition singled out by $R$, since $\delta_k \geq \widehat{P}(\rho)-P_{{\rm LB}_k}$ leaves $\omega_3(\mu_e)$ feasible during the rank reduction step.
By App.~\ref{app:ProofCongNonExtHier}, every limit point as $k\to \infty$ corresponds to the state $\omega_3$ of an optimal decomposition, and for almost every $R$ the minimizer of $\tr(R\,\cdot)$ on this set is unique and has linearly independent projectors.
Then, Jennrich's algorithm is continuous at $\omega_3(\mu_e)$, and therefore applied to $\Omega_3$ it gives a decomposition of a state $\sigma$ that is $\epsilon$-close to $\rho$ with $\epsilon\to0$ as $k\to\infty$, which after the repair step has value converging to $\widehat{P}(\rho)$.

This argument does not extend to LOCC since the optimal decompositions do not form a face of the set of decompositions of $\rho$, and the extreme points may not have linearly independent projectors, in which case we cannot use Jennrich's algorithm.

\section{Linear program for non-polynomial functions}\label{app:LP}

\subsection{Construction of the LP and of the single SDP}
We optimize over symmetric polynomials of degree at most $k$. Since polynomials of lower degree can be lifted to degree $k$ (Lemma~\ref{lem:lifting}), we parametrize them as $P(\a)=\sum_{\lambda\vdash k}c_\lambda s_\lambda(\a)$, and the LP variables are the coefficients $\vec c=(c_\lambda)_\lambda$. Given the function $F$, we require $P(\a)\leq F(\a)$ for all spectra $\a$, and enforce this constraint discretely on a grid of spectra $\{\a^g\}_g$, where the values $s_\lambda(\a^g)$ and $F(\a^g)$ are precomputed. By symmetry, it suffices to take grid points with $\a_{g,1}\geq\a_{g,2}\geq\dots\geq\a_{g,d_A}\geq0$ and $\sum_i\a_{g,i}=1$.

Writing $\mu_\lambda(\Phi)=\tr\left((\tfrac{1}{f^\lambda}\Pi^\lambda\otimes\I^B)\Phi\right)$, the objective is $\tr\left((O_P\otimes\I^B)\Omega_k\right)=\sum_\lambda c_\lambda\,\mu_\lambda(\Omega_k)$, and therefore  the bound reads
\begin{equation}
    F_{{\rm LB}_k}=\max_{\vec c}\ \min_{\Omega_k}\ \sum_\lambda c_\lambda\,\mu_\lambda(\Omega_k)
    \quad\text{s.t.}\quad \sum_\lambda c_\lambda s_\lambda(\a^g)\leq F(\a^g)\ \ \forall g,
\end{equation}
with $\Omega_k$ ranging over the feasible set of Eq.~\eqref{eq:SDPcroof1}. Since the objective is bilinear, the set of feasible $\Omega_k$ is convex and compact, and the set of valid $\vec{c}$ is convex, we can exchange the maximum and the minimum~\cite{Sion1958}. Then, for fixed $\Omega_k$, the inner maximization is an LP whose dual is
\begin{equation}
    \min_{y\geq0}\ \sum_g y_g F(\a^g)\quad\text{s.t.}\quad \sum_g y_g\,s_\lambda(\a^g)=\mu_\lambda(\Omega_k)\ \ \forall\lambda .
\end{equation}
The primal is always feasible, since the constant polynomial $P=\min_g F(\a^g)$ is admissible, so strong duality holds. We merge the LP and the SDP to obtain a single SDP,
\begin{align}
    F_{{\rm LB}_k}=\min_{\Omega_k,\,y}\ & \sum_g y_g F(\a^g)\\
    \text{s.t. }& \Omega_k\ \text{feasible for Eq.~\eqref{eq:SDPcroof1}},\quad y\geq0, \nonumber\\
    & \sum_g y_g\,s_\lambda(\a^g)=\tr\left((\tfrac{1}{f^\lambda}\Pi^\lambda\otimes\I^B)\,\Omega_k\right)\quad\forall\lambda\vdash k. \nonumber
\end{align}
The coefficients $c_\lambda$ of the optimal polynomial are recovered as the dual variables of the equality constraints. Since the grid is discrete, the optimal polynomial $P^*$ may exceed $F$ between grid points. We therefore report $F_{{\rm LB}_k}-\epsilon$, where $\epsilon=\max_\a\left(P^*(\a)-F(\a)\right)$ is estimated from the polynomial returned by the solver. This is a valid lower bound, since $P^*-\epsilon\leq F$ on all spectra.

\subsection{Convergence of the bound}\label{app:LPconvergence}
Here we prove Result~\ref{thm:non-polyHier} for the bound $\widetilde F_{{\rm LB}_k}$ of Eq.~\eqref{eq:FLBexact}, in which $P\leq F$ is imposed on all spectra. Monotonicity follows naturally from the facts that admissible polynomials stay admissible at higher level by lifting (Lemma~\ref{lem:lifting}), and that the SDP is monotone (Appendix~\ref{app:monotonicity}).

By the Weierstrass approximation theorem, for any continuous function $F'$ and any $\eta>0$, we can find a polynomial $P'$ such that $|F'(\a)-P'(\a)| < \eta$ for all spectra $\a$. Then, by averaging over permutations of the variables, for any continuous symmetric function $F$, we can find a symmetric polynomial $P$ such that $|F(\a)-P(\a)| < \eta$ for all spectra $\a$.
Shifting $P$ by $\eta$, for any $F$ and any $\eta>0$, there exists a symmetric polynomial $P$ such that $P(\a) < F(\a)$ and $F(\a)-P(\a) < 2\eta$ for all spectra $\a$. This $P$ is admissible for $\deg P \leq k$, and therefore  by Result~\ref{thm:ConvergenceSDPcroof}, $\lim_{k\to\infty}\widetilde F_{{\rm LB}_k}(\rho)\geq\lim_{k\to\infty}P_{{\rm LB}_k}(\rho)=\widehat P(\rho)\geq\widehat F(\rho)-2\eta$. Conversely, $P\leq F$ implies $\widehat P(\rho)\leq\widehat F(\rho)$ for every admissible $P$, and therefore $\widetilde F_{{\rm LB}_k}(\rho)\leq\widehat F(\rho)$. Since $\eta$ is arbitrary,
\begin{equation}
    \lim_{k\to\infty}\widetilde F_{{\rm LB}_k}(\rho)=\widehat F(\rho).
\end{equation}

With a finite grid, the set of admissible polynomials is larger, and therefore $F_{{\rm LB}_k}\geq\widetilde F_{{\rm LB}_k}$, and $F_{{\rm LB}_k}-\epsilon\leq\widehat F(\rho)$, where $\epsilon$ is the grid error of the returned polynomial. The uncorrected value $F_{{\rm LB}_k}$ is nondecreasing in the level by the same argument as above, but the corrected value $F_{{\rm LB}_k}-\epsilon$ need not be, since the grid error can change between levels. In our numerical results, we observe that increasing the level yields tighter bounds.

\section{LOCC feasibility SDP}\label{app:LOCC}
\subsection{Construction of the SDP}
Let $\ket{\psi}$ be the initial pure state and $\rho$ the target mixed state. By Ref.~\cite{Jonathan_1999}, $\ket{\psi}$ can be converted into $\rho$ by LOCC if and only if there is a decomposition $\{p_i,\ket{\psi_i}\}$ of $\rho$ such that
\begin{equation}\label{eq:LOCCcriterion}
    \sum_i p_i E_j(\psi_i)\leq E_j(\psi)\qquad\text{for all } j=1,\dots,d_A-1 .
\end{equation}
Since all these inequalities must hold for the same decomposition, testing them one monotone at a time is not sufficient. Instead, we combine the LP construction of Appendix~\ref{app:LP} for all the Vidal monotones into a single SDP.

For each $j$, let $\mathcal{P}_j$ be the set of symmetric polynomials $P_j$ of degree at most $k$ that satisfy $P_j(\a^g)\leq E_j(\a^g)$ on all grid points, and let
\begin{equation}
    \epsilon_j(P_j)=\max_{\a}\left(P_j(\a)-E_j(\a)\right)
\end{equation}
be the maximal violation of $P_j$ between grid points. If a decomposition $\{p_i,\ket{\psi_i}\}$ satisfies Eq.~\eqref{eq:LOCCcriterion}, then its operator $\omega_{k}$ is feasible for Eq.~\eqref{eq:SDPcroof1} and, for every $P_j\in\mathcal{P}_j$,
\begin{equation}\label{eq:LOCCnecessary}
\tr\left((O_{P_j}\otimes\I^B)\,\omega_{k}\right)=\sum_ip_iP_j(\a_i)\leq\sum_ip_iE_j(\a_i)+\epsilon_j(P_j)\leq E_j(\psi)+\epsilon_j(P_j),
\end{equation}
where $\a_i$ is the entanglement spectrum of $\ket{\psi_i}$. Therefore, if a certain decomposition can be transformed from $\ket{\psi}$ with LOCC, then it can violate the $j$'th inequality by at most the grid error $\epsilon_j(P_j)$ for each polynomial.

With the moments $\mu_\lambda(\Phi)=\tr\left((\tfrac{1}{f^\lambda}\Pi^\lambda\otimes\I^B)\Phi\right)$ of Appendix~\ref{app:LP}, the optimization of Eq.~\eqref{eq:LOCCminmaxt} is computed by the SDP
\begin{align}\label{eq:LOCCsdp}
    s_k=\min_{\Omega_k,\,y^{(j)},\,s}\ & s\\
    \text{s.t. }& \Omega_k\ \text{feasible for Eq.~\eqref{eq:SDPcroof1}},\quad y^{(j)}\geq0\ \ \forall j, \nonumber\\
    & \sum_g y^{(j)}_g\,s_\lambda(\a^g)=\mu_\lambda(\Omega_k)\quad\forall\lambda\vdash k,\ \forall j, \nonumber\\
    & \sum_g y^{(j)}_g\,E_j(\a^g)-E_j(\psi)\leq s\quad\forall j. \nonumber
\end{align}

For fixed $\Omega_k$, the minimization over each $y^{(j)}$ is the dual LP of Appendix~\ref{app:LP} with $F=E_j$, and therefore by strong duality
\begin{equation}\label{eq:LOCCminmax}
    s_k=\min_{\Omega_k}\ \max_{j}\ \max_{P_j\in\mathcal{P}_j}\Big(\tr\left((O_{P_j}\otimes\I^B) \Omega_k\right)-E_j(\psi)\Big).
\end{equation}
Thus $s_k$ is the smallest worst-case violation of Eq.~\eqref{eq:LOCCcriterion} over the relaxations of decompositions of $\rho$, where each violation is measured with the best grid-admissible polynomial. For a single monotone, $s_k=F_{{\rm LB}_k}-E_j(\psi)$ with $F=E_j$.

Let $t$ be a probability vector with $t_j \geq 0$ and $\sum_j t_j = 1$. Then,
\begin{align}\label{eq:LOCCminmaxT}
    s_k=\min_{\Omega_k} \max_t \max_{P_j\in\mathcal{P}_j}\ \sum_j t_j\Big(\tr\left((O_{P_j}\otimes\I^B)\Omega_k\right)-E_j(\psi)\Big).
\end{align}
With the substitution $Q_j=t_jP_j$, the constraint $P_j\in\mathcal{P}_j$ becomes $Q_j(\a^g)\leq t_jE_j(\a^g)$ for all $g$, and the objective is linear in $\Omega_k$ and in $(t,Q_j)$. The feasible set of $\Omega_k$ is convex and compact, and the set of $(t,Q_j)$ is convex. Again, here the minimum and the maximum can be exchanged by Ref.~\cite{Sion1958}:
\begin{align}
    s_k&=\min_{\Omega_k}\ \max_{t,\,Q_j}\ \sum_j \Big(\tr\left((O_{Q_j}\otimes\I^B)\Omega_k\right)-t_jE_j(\psi)\Big)\nonumber\\
    &=\max_{t,\,Q_j}\ \min_{\Omega_k}\ \sum_j \Big(\tr\left((O_{Q_j}\otimes\I^B)\Omega_k\right)-t_jE_j(\psi)\Big).\label{eq:LOCCmaxmin}
\end{align}
Let $t^*$ and $Q_j^*$ be the optimal vector $t$ and optimal polynomials, which are given by the dual variables of the SDP~\eqref{eq:LOCCsdp}; for $t^*_j>0$, $Q^*_j=t^*_jP^*_j$ with $P^*_j\in\mathcal{P}_j$.
By Eq.~\eqref{eq:LOCCmaxmin}, they satisfy
\begin{equation}\label{eq:LOCCcertificateIneq}
    \sum_j \Big(\tr\left((O_{Q^*_j}\otimes\I^B)\Omega_k\right)-t^*_jE_j(\psi)\Big)\geq s_k
\end{equation}
for every feasible $\Omega_k$.

Let $\epsilon_j=\max_\a\big(Q^*_j(\a)-t^*_jE_j(\a)\big)$, which equals $t^*_j\,\epsilon_j(P^*_j)$ for $t^*_j>0$. If $\ket{\psi}$ can be converted into $\rho$ by LOCC, some decomposition satisfies Eq.~\eqref{eq:LOCCcriterion}, and its $\omega_k$ is feasible. As in Eq.~\eqref{eq:LOCCnecessary}, $\tr\big((O_{Q^*_j}\otimes\I^B)\,\omega_k\big)=\sum_ip_iQ^*_j(\a_i) \leq t^*_jE_j(\psi)+\epsilon_j$, and therefore Eq.~\eqref{eq:LOCCcertificateIneq} gives
\begin{equation}
    s_k \leq \sum_j \Big(\tr\left((O_{Q^*_j}\otimes\I^B)\,\omega_k \right)-t^*_jE_j(\psi)\Big) \leq \sum_j \epsilon_j.
\end{equation}
Thus, if $s_k > \sum_j \epsilon_j$, $\ket{\psi}$ cannot be converted into $\rho$ by LOCC. If the condition does not hold, the test does not tell us anything, in which case, for $k\geq3$, we can extract the decomposition from $\Omega_k$ and repair it as in Sec.~\ref{sec:extractingTheDecomposition}, and if the criterion Eq.~\eqref{eq:LOCCcriterion} holds, the state can be converted into $\rho$ by LOCC.
In this procedure, the grid errors $\epsilon_j$ need to be computed in a postprocessing step from the polynomials returned by the solver.

\subsection{Detection at finite level}
We impose $P_j\leq E_j$ on all spectra.
Consider the decomposition $\{p_i, \dyad{\varphi_i}\}$ of $\rho$ that minimizes
\begin{equation}
    s^*=\max_t \sum_j t_j \left( \left(\sum_i p_i E_j(\varphi_i)\right)  - E_j(\psi) \right)
\end{equation}
over probability vectors $t$.
If $\ket{\psi}\overset{\text{LOCC}}{\nrightarrow}\rho$, at least one Vidal monotone $E_j$ will be exceeded by the decomposition and setting $t_j=1$, we get that $s^*>0$.
We know $\widetilde{s}_k \leq s^*$ since we are relaxing the set of feasible $\omega_k$.
By App.~\ref{app:ProofCongNonExtHier}, as $k\to\infty$, the few-copy marginals of $\Omega_k$ approach mixtures $\sum q_i\dyad{\phi_i}^{\otimes n}$ with $\{q_i, \dyad{\phi_i}\}$ a decomposition of $\rho$, and since by App.~\ref{app:LPconvergence} each $E_j$ can be approximated to arbitrary precision from below, $\widetilde{s}_k \to s^*$ as $k \to \infty$ and every non-convertible pair is eventually detected.

\section{Further analysis}
In this appendix, we provide additional numerical results for the entanglement of formation (Fig.~\ref{fig:EoF}), the linear entropy (Fig.~\ref{fig:namedLinearEntropy}) and the Vidal monotones (Fig.~\ref{fig:Vidal}).
We also show how the LP adapts to different states by choosing different lower-bounding polynomials, which are tangent to the reference function at different points of the spectrum (Fig.~\ref{fig:polynomials}).

\begin{figure*}
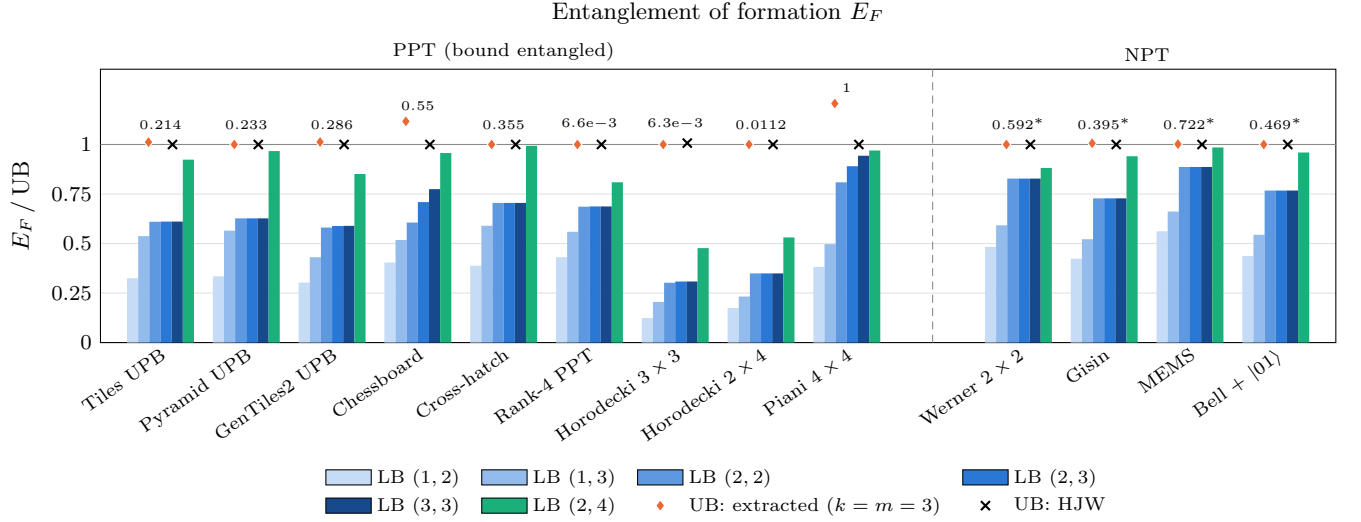

  \centering
  \namedbars{figures/data/named_EF.csv}{$E_{F}\,/\,\mathrm{UB}$}{Entanglement of formation $E_F$}
    \caption{Lower bounds on the entanglement of formation $E_F$ for different PPT and NPT states, from the SDP with embedded LP. Bars: lower bounds from the simplified hierarchy at the levels $(k,m)$ shown in the legend. Markers: upper bounds from the decomposition extracted from the SDP solution at $k=m=3$ (orange diamonds) and from HJW gradient descent over decompositions (black crosses). All values are divided by the best upper bound UB, printed above each state. Here, UB is the exact convex roof for two qubits (marked $^*$), and the lowest of the two upper bounds otherwise. Notice that the gap between upper and lower bounds decreases with the hierarchy level.}
  \label{fig:EoF}
\end{figure*}
\begin{figure*}
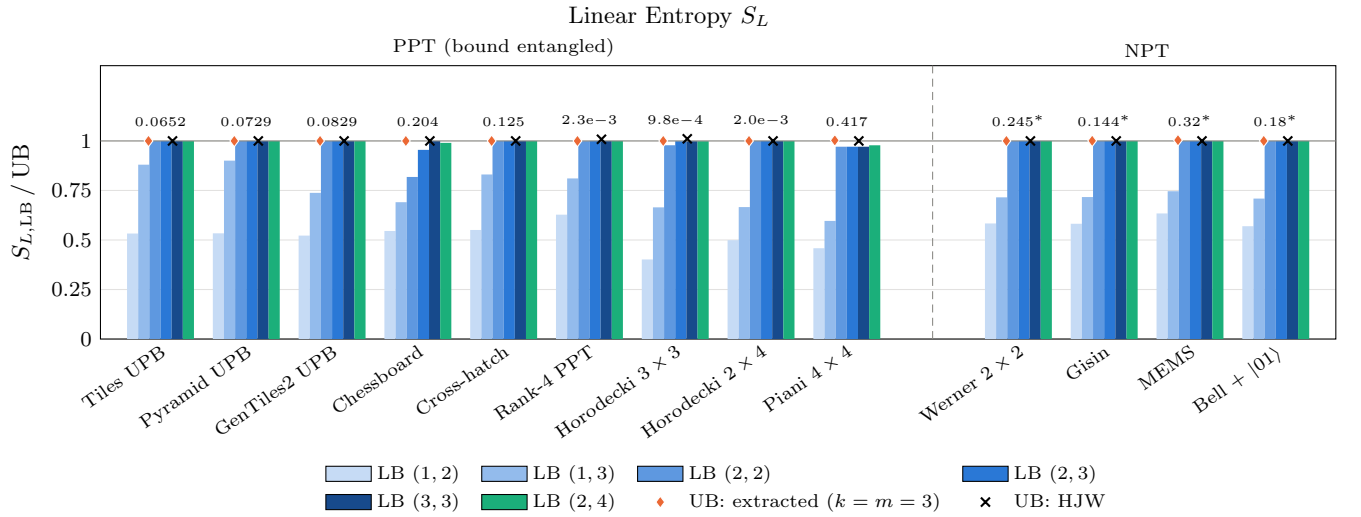

{Linear Entropy $S_L$}
  \centering
  \namedbars{figures/data/named_SL.csv}{$S_{L,\mathrm{LB}}\,/\,\mathrm{UB}$}{}
  \caption{Lower bounds on the linear entropy for different PPT and NPT states. Bars: lower bounds from the simplified hierarchy at the levels $(k,m)$ shown in the legend. Markers: upper bounds from the decomposition extracted from the SDP solution at $k=m=3$ (orange diamonds) and from HJW gradient descent over decompositions (black crosses). All values are divided by the best upper bound UB, printed above each state. Here, UB is the exact convex roof for two qubits (marked $^*$), and the lowest of the two upper bounds otherwise.}
  \label{fig:namedLinearEntropy}
\end{figure*}

\begin{figure*}
  \centering
  \namedbars{figures/data/named_E1.csv}{$E_{1,\mathrm{LB}}\,/\,\mathrm{UB}$}{Vidal monotone $E_1=1-\alpha_1$}
  \caption{Lower bounds on the Vidal monotone $E_1$ for different PPT and NPT states, from the SDP with embedded LP. Bars: lower bounds from the simplified hierarchy at the levels $(k,m)$ shown in the legend. Markers: upper bounds from the decomposition extracted from the SDP solution at $k=m=3$ (orange diamonds) and from HJW gradient descent over decompositions (black crosses). All values are divided by the best upper bound UB, printed above each state. Here, UB is the exact convex roof for two qubits (marked $^*$), and the lowest of the two upper bounds otherwise.}
  \label{fig:Vidal}
\end{figure*}

%
% \polypanel{measure column prefix (EF or E1)}: the optimal polynomial of each state, and the points where it touches F
% (the Schmidt coefficients of the Wootters decomposition).
\providecommand{\polypanel}[1]{%
  \addplot[cA, thick] table[x=alpha, y=werner2] {figures/data/polynomials_#1.csv};\addlegendentry{Werner $2\times2$}
  \addplot[cA, only marks, mark=*, mark size=1.8pt, forget plot]
    table[x=alpha_werner2, y=#1_werner2] {figures/data/polynomials_touching.csv};
  \addplot[cB, thick] table[x=alpha, y=gisin] {figures/data/polynomials_#1.csv};\addlegendentry{Gisin}
  \addplot[cB, only marks, mark=*, mark size=1.8pt, forget plot]
    table[x=alpha_gisin, y=#1_gisin] {figures/data/polynomials_touching.csv};
  \addplot[cC, thick] table[x=alpha, y=mems] {figures/data/polynomials_#1.csv};\addlegendentry{MEMS}
  \addplot[cC, only marks, mark=*, mark size=1.8pt, forget plot]
    table[x=alpha_mems, y=#1_mems] {figures/data/polynomials_touching.csv};
  \addplot[cD, thick] table[x=alpha, y=bellprod] {figures/data/polynomials_#1.csv};\addlegendentry{Bell + $\ket{01}$}
  \addplot[cD, only marks, mark=*, mark size=1.8pt, forget plot]
    table[x=alpha_bellprod, y=#1_bellprod] {figures/data/polynomials_touching.csv};}%
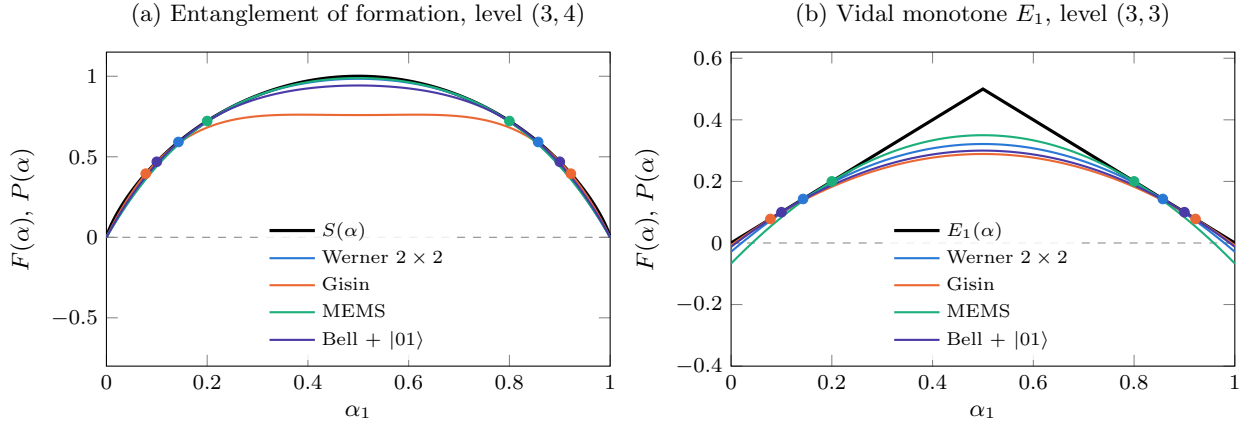
\begin{figure*}
  \centering
  \begin{tikzpicture}
    \begin{groupplot}[group style={group size=2 by 1, horizontal sep=1.6cm}, width=0.46\textwidth,
      height=0.32\textwidth, tick label style={font=\footnotesize}, label style={font=\small}, title style={font=\small},
      xlabel={$\alpha_1$}, ylabel={$F(\alpha)$, $P(\alpha)$}, xmin=0, xmax=1, table/col sep=comma,
      legend style={font=\scriptsize, at={(0.5,0.02)}, anchor=south, draw=none, fill=white, fill opacity=0.85,
      text opacity=1}, legend cell align=left]
    \nextgroupplot[title={(a) Entanglement of formation, level $(3,4)$}, ymin=-0.8, ymax=1.15]
      \addplot[cGray, dashed, forget plot] coordinates {(0,0) (1,0)};
      \addplot[black, very thick] table[x=alpha, y=F] {figures/data/polynomials_EF.csv};\addlegendentry{$S(\alpha)$}
      \polypanel{EF}
    \nextgroupplot[title={(b) Vidal monotone $E_1$, level $(3,3)$}, ymin=-0.4, ymax=0.62]
      \addplot[cGray, dashed, forget plot] coordinates {(0,0) (1,0)};
      \addplot[black, very thick] table[x=alpha, y=F] {figures/data/polynomials_E1.csv};\addlegendentry{$E_1(\alpha)$}
      \polypanel{E1}
    \end{groupplot}
  \end{tikzpicture}
  \caption{
  Polynomials $P\leq F$ extracted from the SDP with embedded LP for different $2$-qubit states. For each state, this is the polynomial that yields the best bound on the convex roof on the grid, at level $(3,4)$ for $E_F$ (a) and at level $(3,3)$ for $E_1$ (b). Here we show
  that the LP is adaptive depending on the state and there is not a single polynomial that yields the best bound on the convex roof for all states at the same level. The states used are $\{$ Werner $2\times2$: singlet weight $p=0.8$; Gisin: $\lambda=0.9$, $\theta=\pi/8$; MEMS: $C=0.8$; Bell + $\ket{01}$: $p=0.6$ $\}$. The polynomials are shown as a function of the entanglement spectrum $\a = (\a_1, \a_2) = (\a_1, 1-\a_1)$.}
  \label{fig:polynomials}
\end{figure*}

\section{Numerical implementation}\label{app:numerics}
All the SDPs are solved with the interior-point conic solver QICS~\cite{qics}, with its default tolerances.
The PPT entangled states from the figures come from the library \texttt{bound-entangled}~\cite{boundEntangled} and the NPT ones are built with \texttt{toqito}~\cite{toqito}. For the Horodecki $3\times3$ and $2\times4$ states we take the parameter $a=0.5$ in Figs.~\ref{fig:EoF}, \ref{fig:namedLinearEntropy} and~\ref{fig:Vidal}. The parameters of the NPT states are the ones given in the caption of Fig.~\ref{fig:polynomials}.

\subsection{Upper bounds by gradient descent}\label{app:HJW}
By the HJW theorem~\cite{HJW}, every decomposition of $\rho=\sum_i\lambda_i\dyad{e_i}$
into $M$ pure states is obtained
by mixing the vectors $\sqrt{\lambda_i}\ket{e_i}$ with an $M\times M$ unitary. Therefore every unitary provides a decomposition of
$\rho$ and hence an upper bound on the convex roof. As in Ref.~\cite{Rothlisberger2009}, we minimize this upper bound by
gradient descent over the unitaries, with $M=r^2$ and several random starting points, and keep the best value.

%
% \memcurve{state}{color}{mark}{legend}: the measured memory of one state (column mem_<state> of memory.csv), broken
% between the groups of levels of equal degree m (pos 1-2: m = 2, pos 4-6: m = 3, pos 8-11: m = 4).
\providecommand{\memcurve}[4]{%
  \addplot[#2, thick, mark=#3, mark size=1pt, restrict x to domain=0:2.5] table[x=pos, y=mem_#1] {\memfile};
  \addlegendentry{#4}
  \addplot[#2, thick, mark=#3, mark size=1.8pt, restrict x to domain=3.5:6.5, forget plot]
    table[x=pos, y=mem_#1] {\memfile};
  \addplot[#2, thick, mark=#3, mark size=1.8pt, restrict x to domain=7.5:11.5, forget plot]
    table[x=pos, y=mem_#1] {\memfile};}%
\providecommand{\fulllevel}[1]{\fill[cGray!15] (axis cs:#1-0.4,0.01) rectangle (axis cs:#1+0.4,32);}%
\begin{figure}
  \centering
  \def\memfile{figures/data/memory.csv}%
  \begin{tikzpicture}
    \begin{axis}[name=memaxis, width=0.85\columnwidth, height=0.6\columnwidth, ymode=log, ymin=0.01, ymax=32, xmin=0.4, xmax=11.6,
      xtick=data, xticklabels from table={figures/data/memory.csv}{level}, xlabel={level $(k,m)$},
      ytick={0.01,0.1,1,10}, yticklabels={10\,MB,100\,MB,1\,GB,10\,GB}, ylabel={memory},
      xticklabel style={rotate=45, anchor=north east, font=\scriptsize, inner sep=1pt},
      tick label style={font=\footnotesize}, label style={font=\small}, xtick style={draw=none},
      ymajorgrids, grid style={cGray!25}, clip=false, table/col sep=comma,
      legend to name=memlegend,
      legend style={font=\scriptsize, draw=none, legend columns=3,
      /tikz/every even column/.append style={column sep=4pt}}, legend cell align=left]
      \fulllevel{2}\fulllevel{6}\fulllevel{11}
      \addplot[draw=none, forget plot] table[x=pos, y expr=1] {\memfile};% the ticks (xtick=data)
      \memcurve{r4a2}{cA}{*}{$r=4$, $r_A=2$}
      \memcurve{r4a3}{cB}{square*}{$r=4$, $r_A=3$}
      \memcurve{r9a3}{cC}{triangle*}{$r=9$, $r_A=3$}
      \memcurve{r4a4}{cD}{diamond*}{$r=4$, $r_A=4$}
      \memcurve{r8a4}{cE}{pentagon*}{$r=8$, $r_A=4$}
      \memcurve{r16a4}{cF}{star}{$r=16$, $r_A=4$}      \node[font=\scriptsize, anchor=south] at (axis cs:1.5,32) {$m=2$};
      \node[font=\scriptsize, anchor=south] at (axis cs:5,32) {$m=3$};
      \node[font=\scriptsize, anchor=south] at (axis cs:9.5,32) {$m=4$};
    \end{axis}
    \node[anchor=north, yshift=-2pt] at (memaxis.outer south) {\pgfplotslegendfromname{memlegend}};
  \end{tikzpicture}
  \caption{Memory needed by the SDP~\eqref{eq:marginals_SDP} with the QICS solver~\cite{qics} at different levels for random states $\rho$ with $r=\operatorname{rank}(\rho)$ and $r_A = \operatorname{rank}(\tr_B(\rho))$. When $k=m$, the full SDP~\eqref{eq:SDPcroof1} is used.}
  \label{fig:memory}
\end{figure}
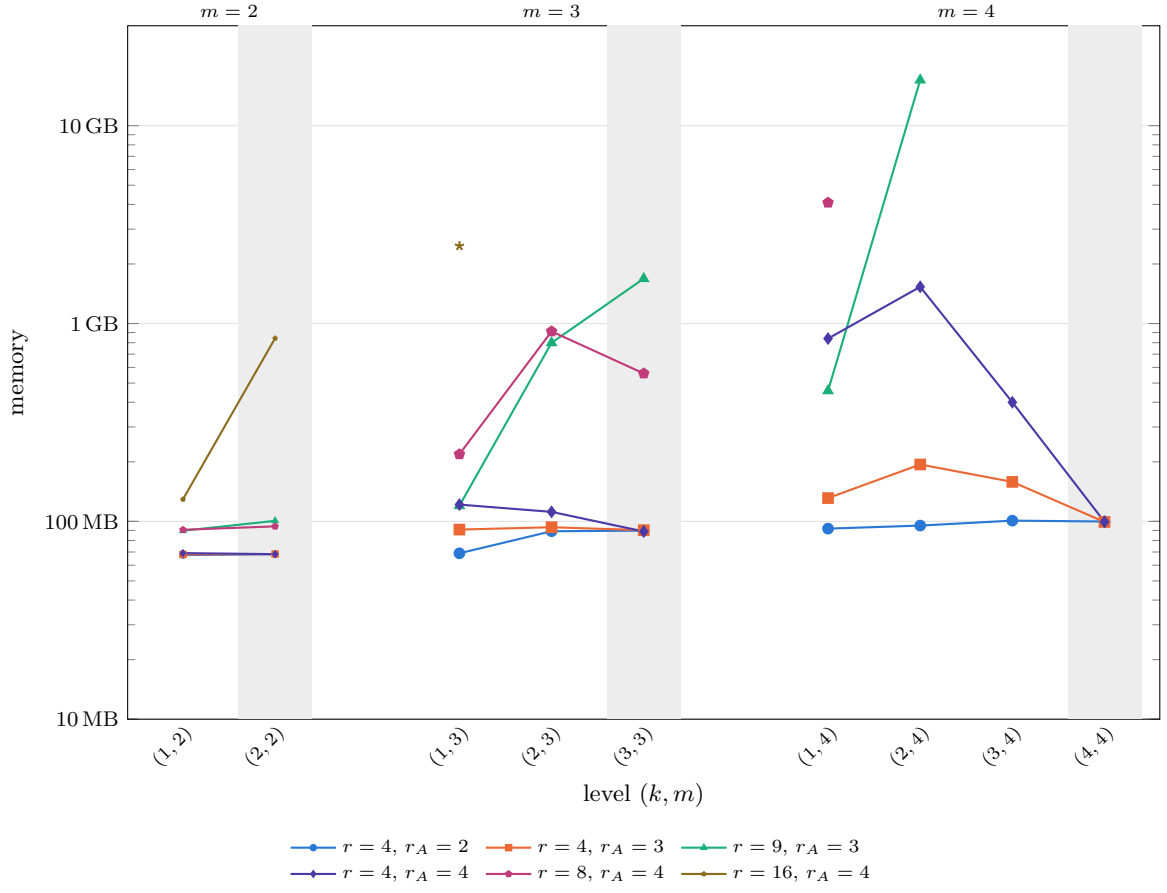

\end{widetext}

\end{document}